\documentclass{article}

\usepackage{microtype}
\usepackage{graphicx}
\usepackage{subfigure}
\usepackage{booktabs} 

\usepackage{hyperref}

\usepackage[accepted]{icml2023}

\usepackage{amsmath}
\usepackage{amssymb}
\usepackage{mathtools}
\usepackage{amsthm}

\usepackage{bm}
\usepackage{nicefrac}

\usepackage[capitalize,noabbrev]{cleveref}

\newtheorem{theorem}{Theorem}[section] 

\newtheorem{remark}{Remark}

\newenvironment{talign*}
 {\csname align*\endcsname}
 {\endalign}
\newenvironment{talign}
 {\csname align\endcsname}
 {\endalign}
\crefname{talign}{}{}
\crefname{equation}{}{}

\usepackage{enumitem} 

\usepackage[acronym]{glossaries}
\newacronym{RKHS}{RKHS}{reproducing kernel Hilbert space}
\newacronym{MC}{MC}{Monte Carlo}
\newacronym{MCMC}{MCMC}{Markov chain Monte Carlo}
\newacronym{CLT}{CLT}{central limit theorem}
\newacronym{CF}{CF}{\emph{control functionals}}
\newacronym{IID}{IID}{independent and identically distributed}
\newacronym{CV}{CV}{\emph{control variate}}
\newacronym{vvCV}{vv-CV}{\emph{vector-valued control variate}}
\newacronym{vvRKHS}{vv-RKHS}{\emph{reproducing kernel Hilbert space of vector-valued functions}}
\newacronym{mvkernel}{mv-kernel}{\emph{matrix-valued reproducing kernel}}
\newacronym{vvfunctions}{vv-functions}{\emph{vector-valued functions}}
\newacronym{TI}{TI}{\emph{thermodynamic integration}}
\newacronym{BQ}{BQ}{Bayesian quadrature}

\setkeys{glslink}{hyper=false}

\DeclareMathOperator*{\argmin}{arg\,min}

\DeclareMathOperator*{\argmax}{arg\,max}

\def\S{\mathcal{S}}
\def\Svv{\mathcal{S}^{\text{vv}}}
\def\Ssv{\mathcal{S}^{\text{sv}}}
\def\Jvv{J^{\text{vv}}}

\def\D{\mathcal{D}}

\def\G{\mathcal{G}}
\def\V{\text{Var}}
\def\U{\mathcal{U}}
\def\H{\mathcal{H}}
\def\Lvv{L^{\text{vv}}}
\def\L{\mathcal{L}}

\def\N{\mathbb{N}}
\def\Nplus{\mathbb{N}_+}

\def\R{\mathbb{R}}
\def\X{\mathbb{R}^d}
\def\O{\mathcal{O}}

\def\EST{\hat{\Pi}}
\def\MC{\hat{\Pi}^{\text{MC}}}
\def\CV{\hat{\Pi}^{\text{CV}}}

\newcommand{\metric}[2]{\left< #1, #2 \right>} 
\def\defn{\equiv}

\usepackage[textsize=tiny]{todonotes}

\icmltitlerunning{Vector-Valued Control Variates}

\begin{document}

\twocolumn[
\icmltitle{Vector-Valued Control Variates}




\begin{icmlauthorlist}
\icmlauthor{Zhuo Sun}{ucl}
\icmlauthor{Alessandro Barp}{cam,ati}
\icmlauthor{Fran\c{c}ois-Xavier Briol}{ucl,ati}
\end{icmlauthorlist}

\icmlaffiliation{ucl}{University College London, London, UK}
\icmlaffiliation{cam}{University of Cambridge, Cambridge, UK}
\icmlaffiliation{ati}{The Alan Turing Institute, London, UK}

\icmlcorrespondingauthor{Fran\c{c}ois-Xavier Briol}{f.briol@ucl.ac.uk}

\icmlkeywords{control variates, kernel methods, Monte Carlo integration, ICML}

\vskip 0.3in
]



\printAffiliationsAndNotice{}  

\begin{abstract}
Control variates are variance reduction tools for Monte Carlo estimators. They can provide significant variance reduction, but usually require a large number of samples, which can be prohibitive when sampling or evaluating the integrand is computationally expensive. Furthermore, there are many scenarios where we need to compute multiple related integrals simultaneously or sequentially, which can further exacerbate computational costs. In this paper, we propose \emph{vector-valued control variates}, an extension of control variates which can be used to reduce the variance of multiple Monte Carlo estimators \emph{jointly}. This allows for the transfer of information across integration tasks, and hence reduces the need for a large number of samples. We focus on control variates based on kernel interpolants and our novel construction is obtained through a generalised Stein identity and the development of novel matrix-valued Stein reproducing kernels. We demonstrate our methodology on a range of problems including multifidelity modelling, Bayesian inference for dynamical systems, and model evidence computation through thermodynamic integration.
\end{abstract}

\section{Introduction} 
 A significant computational challenge in statistics and machine learning is the approximation of intractable integrals. Examples include the computation of posterior moments, the model evidence (or marginal likelihood), Bayes factors, or integrating out latent variables. This challenge has lead to the development of a wide range of \gls{MC} methods; see \cite{Green2015} for a review. Let $f:\X\rightarrow \R$ denote some integrand of interest, and $\Pi$ some distribution with Lebesgue density $\pi$ known up to an intractable normalisation constant. The integration task we consider can be expressed as estimating 
 \begin{talign*}
    \Pi[f] := \int_{\X} f(x) \pi(x) \mathrm{d}x 
 \end{talign*}
 using evaluations of the integrand at some points in the domain: $\{x_i,f(x_i)\}_{i=1}^n$. These evaluations are usually combined to create an estimate of $\Pi[f]$ of the form $\EST[f] =  \frac{1}{n} \sum_{i=1}^n f(x_i)$. For example, when realisations are \gls{IID}, this corresponds to a \gls{MC} estimator. 
In that case, assuming that $f$ is square-integrable with respect to $\Pi$ (i.e. $ \Pi[f^2] < \infty$), we can use the \gls{CLT} to show that such estimators converge to $\Pi[f]$ as $n\rightarrow \infty$, and this convergence is then controlled by the asymptotic variance of the integrand $f$. Analogous results can also be obtained for \gls{MCMC} realisations \citep{Jones2004}, in which case $\{x_i\}_{i=1}^n$ are realisations from a Markov Chain with invariant distribution $\Pi$, or for randomised quasi-Monte Carlo \citep{Hickernell2005}, in which case $\{x_i\}_{i=1}^n$ form some lattice or sequence filling some hypercube domain.

The main insight behind the concept of \gls{CV} is that it is instead often possible to use an estimator of $\Pi[f-g]$ for some $g:\X \rightarrow \R$. This is justified if $\Pi[g]$ is known in closed form, in which case we may use $\CV[f] := \EST[f-g] +\Pi[g]$. 
Furthermore, if $g$ is chosen appropriately, the variance of the CLT for this new estimator will be much smaller than that of the original, and a smaller number of samples will be required to approximate $\Pi[f]$ at a given level of accuracy.

Suppose now that $n = (n_1,\ldots, n_T) \in \N^T$ ($T \in \Nplus$) is a multi-index. In this paper, we will focus on cases where we have not just one integral, but a sequence of integrands $f_t:\X\rightarrow \R$ and distributions $\Pi_t$ for which we would like to use $\{\{x_{tj},f_t(x_{tj})\}_{j=1}^{n_t}\}_{t=1}^T$ to estimate
\begin{talign}
\label{eq:multiple_integration_problem}
    \Pi_t[f_t] := \int_{\X} f_t(x) \pi_t(x) \mathrm{d}x \quad \text{for } t \in [T], 
\end{talign}
where $[T]=\{1, \ldots, T\}$.  This is a common situation in practice; for example, our paper considers the case of multifidelity modelling \citep{Peherstorfer2018} where $f_1,\ldots,f_T$ may be a computationally expensive physical model $f$, and we might be interested in expectations of that model with respect to unknown parameters. Another example we will also study is when $\pi_1,\ldots,\pi_T$ are closely related posterior distributions, such as in the case of power posteriors \citep{friel2008_marginal_llk_by_power_posterior}.

Of course, the estimation of integrals in  \Cref{eq:multiple_integration_problem} could be tackled individually, and this is in fact the most common approach. However, the main insight from this paper is that if both the integrands and distributions are related across tasks, we can improve on this by sharing computation across these tasks. We propose to construct a \gls{CV} to \emph{jointly reduce the variance} of estimators for these integrals and hence obtain a more accurate approximation. We will call such a function a \gls{vvCV}. In order to encode the relationship between integration tasks, we will propose a flexible class of \gls{CV}s based on interpolation in \gls{vvRKHS}. More precisely, we generalise existing constructions of Stein reproducing kernels to derive novel \gls{vvRKHS}s with the property that each output has mean zero.

We note that very few methods exist to tackle multiple integrals jointly. One exception is \cite{ICML2018_BQforMultipleRelatedIntegrals}, which also proposes an algorithm based on \gls{vvRKHS}s. However, that work is limited to cases where the integral of the kernel is known in closed-form, which is rarely possible in practice. In contrast, our \gls{vvCV}s are applicable so long as $\pi_t$ is known up to an unknown constant and $\nabla_x \log \pi_t$ can be evaluated pointwise for all $t \in [T]$ (where $\nabla_x = (\partial/\partial x_1, \ldots, \partial/\partial x_d)^\top$). This will usually be satisfied in Bayesian statistics, and is a requirement for the implementation of most gradient-based \gls{MCMC} algorithms.

The remainder is as follows. In Section \ref{sec:background}, we review existing \gls{CV}s based on Stein's method. In Section \ref{sec:methodology}, we introduce \gls{vvCV}s, then in Section \ref{sec:learning_vvCVs} we show how to find an optimal \gls{vvCV}. Finally, in Section \ref{sec:experiments}, we demonstrate the advantage of our approach on problems in multifidelity modelling, Bayesian inference for differential equations and model evidence computation through thermodynamic integration. 

\vspace{-4mm}

\paragraph{Notation} Vectors $x \in \R^d$ are column vectors, $\|x\|_q = (\sum_{i=1}^d x_i^q)^{1/q}$ for $q \in \N$, and $\pmb{1}_d =(1,\ldots,1)^\top \in \R^d$. For a multi-index $m\in \N^d$, we write $|m|=\sum_{i=1}^d m_i$ for its total degree. For a matrix $M \in \R^{p \times q}$, $M_{ij}$ denotes the entry in row $i$ and column $j$, $\|M\|_F^2 = \sum_{i=1}^p \sum_{j=1}^q M_{ij}^2$ is the Frobenius squared-norm, $\text{Tr}(M) = \sum_{i=1}^m M_{ii}$ is the trace, and $M^\dagger$ is the pseudo-inverse. $I_m$ denotes the $m$-dimensional identity matrix, and $S^m_+$ the set of symmetric strictly positive definite matrices in $\R^{m \times m}$. 
We denote by $C^j$ the set of functions whose mixed partial derivatives of order at most $j$ are continuous, and given a differentiable function $g$ on $\R^{d_1} \times \R^{d_2}$,
$\partial^r_x g(x,y)$ denotes its partial derivative in the $r^{th}$-coordinate of its first entry evaluated at $(x,y)$.


\section{Background}
\label{sec:background}

We now briefly review existing constructions for Stein-based \gls{CV}s for a single integration problem.

The first step consists of constructing a set of functions $\G$ which all integrate to a known value against $\Pi$. This can generally be challenging since $\pi$ may not be computationally tractable, e.g., involving an unknown normalisation constant. Without loss of generality, we will discuss the construction of functions which integrate to zero, but notice that we can obtain functions with mean equal to any constant $\beta \in \R$ by simply adding this constant $\beta$ to a zero-mean function.

\paragraph{Zero-Mean Functions through Stein Operators}

One way of constructing zero-mean functions is to use Stein's method  \citep{Anastasiou2021}. The main ingredients of Stein's method are a function class and an operator acting on this class. More precisely, a \emph{Stein class} of $\Pi$ is a class of functions $\U$ associated to an operator $\S$, called \emph{Stein operator}, such that a \emph{Stein identity} holds: $
\Pi[\S [u]] = 0$ $\forall u \in \U$. An obvious choice for the class of zero-mean functions $\G$ is to consider all functions of the form $g=\S[u]$ for $u\in \U$. To ensure such $g$ has finite variance, we assume that all functions in $\G$ are square-integrable with respect to $\Pi$. This can be guaranteed under weak regularity conditions on $\U$ and $\S$; see \Cref{thm:squared inegrable RKHS}. Note also that $\S$ depends implicitly on $\Pi$, but we only make this explicit in our notation (i. e. $\S_\Pi$)  when it is helpful for clarity. 

The most common choice of Stein operator is the \emph{Langevin Stein operator}, which acts on differentiable
\gls{vvfunctions} $u:\X \rightarrow \R^d$:
\begin{talign}\label{eq:Langevin_vv}
\L[u](x):= \nabla_x \cdot u(x) + u(x) \cdot \nabla_{x} \log \pi(x).
\end{talign}
The advantage of $\L$  is that it only requires knowledge of $\Pi$ through evaluations of $\nabla_x \log \pi$, which does not require the normalisation constant of $\pi$. Indeed, let $\pi = \tilde{\pi}/C$ for some unknown $C \in \R$, then $\nabla_x \log \pi = \nabla_x \log \tilde{\pi}$. For more general Stein operators, see \cite{Anastasiou2021}.

\paragraph{Parametric Spaces} 
$\G$ is usually chosen to be a parametric space and we will hence write it $\G_\Theta$, where $\Theta$ denotes the space of  parameter values.  Most existing \gls{CV}s can be obtained by taking $g_\theta = \S[u_\theta]$ for some $u_{\theta} \in \U_\Theta$ where $\U_\Theta$ is another parametric class. However, note that there might not be a unique $u_\theta$ leading to $g_\theta$. For the remainder of the paper, $\theta$ will hence be a parameter indexing $g_\theta$ directly as opposed to an element of the Stein class. Examples of parametric \gls{CV}s are the polynomial-based \gls{CV}s of \cite{mira2013_zerovariance_MCMC} (see also  \cite{assaraf1999_zerovariance_principleforMCI, papamarkou2014_zerovariance_diffgeoMCMC,oates2014_controlled_themo_integration}), in which case the Stein class is parametrised directly by coefficients of a polynomial. Using a space of neural networks has also been studied in \cite{zhu2018neural_CV,Si2020}. This later choice can be advantageous due to the flexibility of this function class, but is much more challenging to implement  because selecting a \gls{CV} becomes a non-convex problem.

Another example are kernel interpolants, which will be the main focus of our paper. This class is nonparametric, but it is often convenient to fix the dataset size and parametrise it. Let $\H_k$ denote a \gls{RKHS} with kernel $k:\X\times\X \rightarrow \R$ \citep{berlinet2011_RKHSinProbsStats}, so that $k$ is symmetric ($k(x,y) = k(y,x)$ $\forall x,y \in \X$) and positive semi-definite ($\forall m \in \Nplus$, $\sum_{i,j=1}^m c_i c_j k(x_i,x_j) \geq 0$  $\forall c_1,\ldots,c_m \in \R$ and $\forall x_1,\ldots,x_m \in \X$). The kernel could be a squared-exponential kernel $k(x,y) =\exp(-\|x-y\|_2^2/2l^2)$ with lengthscale $l>0$, or a polynomial kernel $k(x,y) = (x^\top y +c)^l$ where $c \in \R$ and $l\in \N$ is the degree of the polynomial. 
\citet{oates2017_CF_for_MonteCarloIntegration}  noticed that the image of $\U = \mathcal{H}_k^d := \mathcal{H}_{k} \times \ldots \times \mathcal{H}_{k}$ under $\L$ is a \gls{RKHS} with kernel
\begin{talign}\label{eq:k_0}
    k_0 (x, y) 
    &:=  \nabla_{x} \cdot \nabla_{y} k(x, y)  + \nabla_{x} \log \pi(x) \cdot \nabla_{y} k(x, y)   \nonumber\\
    &\quad +  \nabla_{y} \log \pi(y) \cdot \nabla_{x} k(x, y) \nonumber\\
    &\quad + ( \nabla_{x} \log \pi(x) \cdot  \nabla_{y} \log \pi(y))  k(x, y),
\end{talign}
see also Thm 2.6 \citet{barp2022targeted} for a more general result. Given $m$ observations, it is known that the optimal interpolant in $\H_{k_0}$ is of the form $g_\theta(x) = \sum_{i=1}^m \theta_i k_0(x,x_i)$ where $\theta_i \in \R$ for all $i \in [m]$. This therefore provides a natural parametrisation for practical implementation. \citet{oates2017_CF_for_MonteCarloIntegration} called this class of \gls{CV}s \gls{CF}; see also  \cite{Briol2017SMCKQ,oates2019convergence,south2020_Semi_Exact_CF} for more details.

\paragraph{Selecting a \gls{CV}} In order to select a \gls{CV}, we will pick the ``best'' element from $\G_\Theta$, where ``best'' will refer to minimising \gls{MC} variance:
\begin{talign}\label{eq:objective_singlecase}
    J(\theta) = \text{Var}_{\Pi}[ f- g_\theta] :=\Pi[(f-g_{\theta} - \Pi[f])^2],
\end{talign}
Following the framework of empirical risk minimisation, this can be approximated with $\{x_j,f(x_j)\}_{j=1}^m$ as follows:
\begin{talign*}
    J_m(\theta,\beta) & = \frac{1}{m} \sum_{j=1}^m (f(x_j) - g_\theta(x_j) - \beta)^2  + \lambda \|g_\theta\|^2,
\end{talign*}
where $\beta \in \R$ is an additional parameter which tends to $\Pi[f]$ as $m \rightarrow \infty$ and $\lambda \rightarrow 0$. Here,  $\lambda \geq 0$ is a regularisation parameter and $\|g_\theta\|$ can be any suitable norm. For example, $\|g_\theta\| = \|\theta\|_2$ or $\|g_\theta\| = \|g_\theta\|_{\H_k}$ for some kernel $k$. Assuming that $\Theta \subseteq \R^p$, this objective can then be minimised by the solution to a linear system when $\theta \mapsto g_\theta$ is linear and $\theta \mapsto \|g_\theta\|^2$ is quadratic. In more general cases, it can be minimised using stochastic optimisation \citep{Si2020}. In that case, we initialise $\theta^{(0)}$ and $\beta^{(0)}$, then iteratively take gradient steps with minibatches of size $\tilde{m} \ll m$.

There are two particular perspectives which motivate the objective in \eqref{eq:objective_singlecase}. Firstly, it can be interpreted as a least-squares objective for the function $f - \Pi[f]$ \citep{leluc2021_cv_selection_for_montecarlo_integration}. Secondly, by noticing that for any square-integrable function $h$ and \gls{MC} estimator with $n$ samples we have  $\text{Var}_\Pi[ \MC[h]] =\text{Var}_\Pi[h]/n $, we can notice that $J(\theta)$ fully determines the \gls{CLT} variance of a \gls{MC} estimator of $\Pi[f-g_\theta]$, which makes it particularly well suited for these estimators. This latter viewpoint has also motivated alternative objectives based on the variance of the randomised quasi-Monte Carlo or \gls{MCMC} \gls{CLT}; see e.g. \citet{Hickernell2005,Oates2016CFQMC,Dellaportas2012,Mijatovic2018}. The \gls{MCMC} case is briefly discussed in \Cref{appendix:alter_obj_MCMC}, but the remainder of the paper will focus on the objective in \eqref{eq:objective_singlecase} for simplicity.

\paragraph{The \gls{CV} Estimator}\label{sec:background3}

Recall that a \gls{CV} estimator takes the form $\CV[f] := \EST[f-g] +\Pi[g]$ for some \gls{CV} $g$. Once a function $g_{\hat{\theta}}$ has been selected through the procedure in the previous subsection, the only part missing is therefore an estimator for $\Pi[f-g_{\hat{\theta}}]$. We present two possible approaches below.
 A first option is to create an estimator based on the remainder of the data $\{x_i,f(x_i)\}_{i=m+1}^{m+n}$:
\begin{talign*}
    \EST[f-g_{\hat{\theta}}]  = \frac{1}{n} \sum_{i=m+1}^{m+n} (f(x_i) - g_{\hat{\theta}}(x_i)). 
\end{talign*}
This estimator is unbiased whenever $\EST[f-g_{\hat{\theta}}]$ is  unbiased. This is the case when using a \gls{MC} estimator, but not when using a \gls{MCMC} estimator. The question of how to select $m$ is of practical importance for the quality of the estimator. 
When $m$ is small, most of the function evaluations are used for the \gls{MC} estimator, whereas when $m$ is large, most of the evaluations are used to construct the \gls{CV}.

Due to the difficulty of choosing $m$, a second option is to use all of the data for selecting $g_{\hat{\theta}}$ (i.e. $n=0$). More precisely, denoting by $(\hat{\theta},\hat{\beta})$ the minimiser of $J_m(\theta,\beta)$, we could use $\EST[f-g_{\hat{\theta}}]  = \hat{\beta}$. This estimator will be biased, but will also be more accurate than the first estimator if the least-squares problem can be solved at a fast rate than the MC CLT.

\section{Methodology}
\label{sec:methodology}

We will now extend existing \gls{CV}s to the case where multiple related integration problems are tackled jointly. This will be done through a multi-task learning approach \citep{micchelli2005_On_Learning_Vectorvalued_functions,evgeniou2005_Learning_MultiTask_with_Kernelmethods}, where each integral corresponds to a task and the relationship between tasks will be modelled explicitly. This will allow us to share information across integration tasks, and hence improve accuracy when the number of integrand evaluations is limited.

\subsection{Vector-Valued Functions using Stein's Method}
\label{sec:general_framework_VecValued_CVs_case1}

For the remainder of this paper, we consider integrands corresponding to the outputs of a vector-valued function $f:\X \rightarrow \R^T$ so that $f(x) = (f_1(x), \ldots, f_T(x))^\top$. Our objective is to approximate $\Pi[f] = (\Pi_1[f_1],\ldots,\Pi_T[f_T])^\top$, and we shall assume that $f_t$ is square-integrable with respect to $\Pi_t$ $\forall t \in [T]$. Formally, $\Pi$ is known as a \emph{vector probability distribution}. 
To approximate $\Pi[f]$, we will construct a class of zero-mean \gls{vvfunctions} through Stein's method. This can be done by considering a Stein class $\U$ whose image $\G$ under the Stein operator $\Svv:\U \rightarrow \G$ is a class of $\R^T$-valued functions on  $\X$. 
We will need a generalised form of Stein identity for \gls{vvfunctions}:
\begin{talign}\label{eq:generalised_Stein_equation}
    \Pi_t[g_t] = \Pi_t[(\Svv[u])_t]=0 ,~~ \forall u \in \U \text{ and } \forall t \in [T].
\end{talign}
In other words, each output of the vv-function should integrate to zero against the corresponding probability distribution. Of course, the ordering of the sequence of integrands and distributions matters here, as we do not guarantee that $\Pi_{t}[g_{t'}]=0$ for $t \neq t'$. The property above can be obtained by constructing an operator $\Svv$ through a sequence of Stein operators $\Ssv_{\Pi_t}$ for $t \in [T]$ whose images are scalar-valued functions integrating to zero under $\Pi_t$. These can then be applied in an element-wise fashion as follows
\begin{talign}\label{eq:vvop_sep}
g = \Svv[u] = (\Ssv_{\Pi_1}[u_1],\ldots,\Ssv_{\Pi_T}[u_T])^\top.
\end{talign}
Once again, $\G$ can be parametrised and we will denote it $\G_\Theta$. We can then use an objective based on the variances $f-g_\theta$ to select an optimal element:
\begin{talign}\label{eq:Jvv}
    \Jvv(\theta)  = \| \V_\Pi[f-g_\theta] \|  &  = \|\Pi[(f - g_\theta - \Pi[f])^2 ]\|,
\end{talign}
where $g_\theta \in \G_\Theta$. In the above $\V_\Pi$ should be thought of as applying $\V_{\Pi_t}$, the variance under $\Pi_t$, to the $t^\text{th}$ element of the vv-function. The norm could be any norm on $\R^T$, but we will usually make use of the $1$-norm so as to interpret this objective as the sum of variances on each integrand.
For this objective to make sense, we require $(g_\theta)_t$ to be squared-integrable with respect to $\Pi_t$ $\forall t \in [T]$. Similarly to the $T=1$ case, we are focusing on a least-squares objective, which directly controls the variance of \gls{MC} estimators, but this could be adapted to other estimators as previously discussed.

Let $m = (m_1,\ldots,m_T) \in \Nplus^T$.
Once again, the objective can be approximated via \gls{MC} estimates based on the dataset $\D=\{\{x_{1j},f_1(x_{1j})\}_{j=1}^{m_1},\ldots,\{x_{Tj},f_T(x_{Tj})\}_{j=1}^{m_T}\}$ following the framework of empirical risk minimisation. For example, when the norm above is a $1$-norm, the objective is simply the sum of individual variances:
\begin{talign}\label{eq:Lvv_m}
   &\Lvv_m(\theta,\beta) :=  \Jvv_m(\theta,\beta) + \lambda \|g_\theta\|^2  \\
    & = \sum_{t=1}^T  \frac{1}{m_t} \sum_{j=1}^{m_t} (f_t(x_{tj}) - (g_\theta(x_{tj}))_t - \beta_t)^2 + \lambda \|g_\theta\|^2,\nonumber
\end{talign}
where $\lambda \geq 0$ and now $\beta = (\beta_{1},\ldots,\beta_{T})\in \R^T$. Once again, the second term is used to regularise $\Jvv_m$, but the norm acts on \gls{vvfunctions}. 
For $\U$, we could take a class of polynomials, kernels or neural networks, but will usually require flexible classes of \gls{vvfunctions} in order to encode relationships between tasks. Assuming that each $\Pi_t$ has a $C^1$ and strictly positive density $\pi_t$ with respect to the Lebesgue measure, we will also be able to use the Langevin Stein operators $\L$, which only require access to the score functions for each distribution $\Pi_1, \ldots, \Pi_T$. For this purpose, we now define $l: \X \rightarrow \R^{T \times d}$ to be the matrix-valued function with entries $l_{ij} (x) = \partial^j \log \pi_i(x)$.

\subsection{Kernel-based Vector-Valued \gls{CV}s}
\label{sec:detailed_construction_VecValued_CVs_Case1}

The main choice of Stein class $\U$ studied in this paper is \gls{vvRKHS}s \citep{carmeli2005_VectorValuedRKHS_MercerTheorem,Carmeli2010,Alvarez2011}. This choice is particularly convenient as it allows us to build on the rich literature in statistical learning theory which considers kernels encoding relationships between tasks. A main contribution of this section will be the design of novel Stein reproducing kernels specifically for numerical integration, a task not commonly tackled in statistical learning theory.

A \gls{vvRKHS} $\H_K$ is a Hilbert space of functions mapping from $\X$ to $\R^T$ with an associated \gls{mvkernel} $K:\X\times \X \rightarrow \R^{T \times T}$ which is symmetric ($K(x,y) = K(y,x)^\top$ $\forall x,y \in \X$) and positive semi-definite ($\forall m \in \Nplus$, $ \sum_{i,j=1}^{m} c_i^\top K(x_i, x_j) c_j \ge 0$ for all $c_1,\ldots,c_m \in \R^T$ and $x_1,\ldots,x_m \in \X$).
Any \gls{vvRKHS} satisfies a reproducing property, so that $\forall f \in \mathcal{H}_K$, $f(x)^\top c = \langle f, K(\cdot, x) c \rangle_{\mathcal{H}_K}$ and $K(\cdot, x)c \in \mathcal{H}_{K}$  $\forall x \in \X$ and $\forall c \in \R^{T}$. The reproducing kernels discussed in \Cref{sec:background} are a special case which can be recovered when $T=1$. 
We say that $K$ is  $C^{r,r}(\X \times \X)$ provided 
 that $\partial^{\alpha}_x \partial^{\alpha}_y K(x,y)$ is continuous for all multi-indices $\alpha = (\alpha_1,\ldots, \alpha_d)$ with $ \alpha_1 + \cdots + \alpha_d  \leq r$, 
and that $K$ is bounded with bounded derivatives if there exists $C\geq 0$ for which  $\| \partial_x^\alpha \partial_y^\alpha K(x,y) \|_F \leq C $ for all $x,y \in \R^d$ and multi-indices $\alpha$ with $\alpha_1+\cdots + \alpha_d  \leq 1$.

To construct \gls{vvCV}s, a natural approach is to construct a \gls{mvkernel} $K_0$ using another \gls{mvkernel} $K$ and an operator $\Svv$. 
Then, assuming we have access to such a $K_0$ and to data $\D$, a natural generalisation of the scalar valued case is:
\begin{talign}\label{eq:closed_form_vvCV}
g_\theta(x) = \sum_{t=1}^{T} \sum_{j=1}^{m_t}  K_0(x,x_{tj}) \theta_{tj},  
\end{talign}
where $\theta_{tj} \in \R^T $, for all $t \in [T], j \in [m_t]$.
The remainder of this section will hence focus on how to obtain $K_0$.
Our construction is based on the \gls{CF}s of \cite{oates2017_CF_for_MonteCarloIntegration,oates2019convergence}, which use a scalar-valued kernel $k_0$ obtained through a tensor product structure $\U = \H_k^d$. The initial motivation for introducing \gls{vvfunctions} in this setting was that $\L$ requires inputs which are \gls{vvfunctions} which can be thought of as introducing a dummy dimension for convenience, and is in no way related to the multitask setting in this paper. We now illustrate the natural extension of $k_0$ to a \gls{mvkernel}.

\begin{theorem}\label{thm:K0_kernel_firstorder}
Consider $\mathcal{H}_K$ which is a \gls{vvRKHS} with \gls{mvkernel} $K:\X \times \X \rightarrow \R^{T \times T}$, and suppose that $K \in C^{1,1}(\X \times \X)$.
Furthermore, assuming $u = (u_1,\ldots,u_T) \in \mathcal{H}_K$, let $
\Svv[u]= (\L_{\Pi_1}[u_1], \ldots, \L_{\Pi_T}[u_T])^\top$.
Then, the image of $\mathcal{H}_K^d = \mathcal{H}_K \times \ldots \times \mathcal{H}_K$ under $\Svv$ is a \gls{vvRKHS} with kernel $K_0:\X \times \X \rightarrow \R^{T \times T}$ with:
\begin{talign*}
&(K_0(x,y))_{tt'} 
 = \sum_{r=1}^d \partial^r_x \partial^r_y K(x,y)_{tt'}+ l_{t'r}(y)\partial^r_x K(x,y)_{tt'} \\
& \qquad \qquad + l_{tr}(x) \partial_y^r K(x,y)_{tt'} + l_{tr}(x) l_{t'r}(y) K(x,y)_{tt'} ,
\end{talign*}
\end{theorem}
\noindent The proof is in \Cref{appendix:proof_K0_general}. Notice that $(K_0(x,y))_{tt'} $ depends only on the score function of $\Pi_t$ and $\Pi_{t'}$, and so we (once again) do not require knowledge of normalisation constants. However, in order to evaluate $K_0(x,y)$ for some $x,y \in \X$, we will require pointwise evaluation of $\nabla_{x} \log \pi_t(x)$ and $\nabla_{y} \log \pi_t(y)$ for all $t \in [T]$. Finally, another interesting point is that $(K_0(x,y))_{tt'}$ is a scalar-valued kernel when $t=t'$, but this is not the case for $t \neq t'$ since it is not symmetric in that case. 

To use elements of this RKHS as \gls{vvCV}s, we will require that the least-squares objective in \Cref{eq:Jvv} is well-defined. This can be guaranteed when the elements of the RKHS are square-integrable, and the theorem below, proved in \Cref{app:proof square integrable RKHS}, provides sufficient conditions for this to hold.
\begin{theorem}\label{thm:squared inegrable RKHS}
Suppose that $K$ is bounded with bounded derivatives, and 
$ \Pi_t[\| \nabla_x \log \pi_t \|_2^2] < \infty$ for all $t \in [T]$. Then, for any $g \in\H_{K_0}$, $g_t$ is square-integrable with respect to $\Pi_t$ for all $t \in [T]$.
\end{theorem}

Now the \gls{mvkernel} in  \Cref{thm:K0_kernel_firstorder} takes a very general form as it has minimal requirements on $K$ or $\Pi_1,\ldots,\Pi_T$. This is convenient as it can be applied in a wide range of settings, but this generality comes at the cost of computational complexity. We will now study several special cases which will often be sufficient for applications.

\paragraph{Special Case I: Separable kernel $K$}  For simplicity, the literature on multi-task learning often focuses on the case of \emph{separable kernels}. We say a \gls{mvkernel} $K$ is separable if it can be written as $K(x,y) = B k(x,y)$, where $k$ is a scalar valued kernel and $B\in S^T_+$. The advantage of this formulation is that it decouples the model for individual outputs (as given by $k$) from the model of their relationship, as given by the components of the matrix $B$, which can be thought of as a covariance matrix for tasks. As we will see in Section \ref{sec:learn_B_jointly_case1}, this can be particularly advantageous for selecting the hyperparameters of \gls{vvCV}s. Using such a kernel $K$, the kernel $K_0$ in \Cref{thm:K0_kernel_firstorder} becomes:
\begin{talign*}
(K_0(x,y))_{tt'} & = B_{tt'} \sum_{r=1}^d \partial^r_x \partial^r_y k(x,y)+ l_{t'r}(y)\partial^r_x k(x,y) \nonumber\\
& \quad + l_{tr}(x) \partial_y^r k(x,y) + l_{tr}(x)l_{t'r}(y) k(x,y) ,
\end{talign*}
for $\forall t,t' \in [T]$. This expression is interesting because it reduces the choice of $K$ to the choice of a matrix $B$ and a kernel $k$, and this matrix $B$ has a natural interpretation in that $B_{tt'}$ denotes the covariance between $f_t$ and $f_t'$. See \Cref{appendix:illustration_mvKernel} for illustrations of such $K_0$.

\paragraph{Special Case II: Separable kernel $K$ with one target distribution} A further simplification of the kernel is possible when $K$ is separable and all distributions are the same: $\Pi_1=\cdots=\Pi_T\defn \Pi$. In this case $K_0$ itself becomes a separable kernel of the form: $
    (K_0(x,y))_{tt'} = B_{tt'} k_0(x,y)$ $\forall t,t' \in [T]$,
where $k_0$ is given in \Cref{eq:k_0}. The scalar case can then be recovered by taking $T=1$ and $B=1$.

\paragraph{Selecting a Base Kernel and Alternative Constructions} The base kernel $K$ needs to be chosen by the user. As for the scalar case, we expect that the performance of our approach will depend on whether the smoothness of $K_0$ closely matches the smoothness of the integrand $f$, and the smoothness of $K$ should therefore be chosen accordingly. We also suggest selecting the hyperparameters of $K$ through either cross validation or through maximisation of the log-marginal likelihood of a zero-mean multi-output Gaussian process with kernel $K_0$; see \Cref{appendix:hyperparameters} for details. Furthermore, see \Cref{appendix:alternative_constructions} for alternative constructions, including \gls{vvCV}s derived from the second-order matrix-valued Stein kernels and \gls{vvCV}s based on other parametric spaces such as polynomials and neural networks.

\section{Selecting a Vector-Valued \gls{CV}}\label{sec:learning_vvCVs}

We will now derive both a closed-form expression for the optimal parameters of these kernel-based \gls{vvCV}s and a stochastic optimisation scheme to approximate it.

\paragraph{Closed-form Solutions} \label{sec:closed_form_vvCV}

We start with a theorem which provides a result akin to the \gls{RKHS} representer theorem, but which focuses specifically on the function which minimises the objective in \Cref{eq:Lvv_m}. This theorem shows that there exists a unique parameter minimising the variance objective. 
\begin{theorem}\label{thm:closed_form_minimiser}
Let $\D=\{\{x_{1j},f_1(x_{1j})\}_{j=1}^{m_1}, \ldots, \{x_{Tj},f_T(x_{Tj})\}_{j=1}^{m_T}\}$. The function which minimises the objective in \Cref{eq:Lvv_m}  where $\|g_\theta\| := \|g_\theta\|_{\H_{K_0}}$ and $\beta \in \R^T$ is of the form:
\begin{talign*}
g_{\theta}(x) =  \sum_{t=1}^{T} \sum_{j=1}^{m_t} \theta_{tj}^\top K_0(x,x_{tj}) ,  \theta_{tj} \in \R^T 
\end{talign*}
with optimal parameter $\theta^*$ given by the solution of this convex linear system of equations:
\begin{talign*}
\small &\sum_{t'} \sum_{j'} \Big( \sum_{t}\frac{1}{m_t}   \sum_{j}  
     K_0(x_{t''j''},x_{tj})_{\cdot t}  K_0(x_{tj},x_{t'j'})_{t \cdot} \\
     & \qquad \qquad + \lambda  K_0(x_{t''j''},x_{t'j'}) \Big) \theta_{t'j'}^* \\
    &=\sum_{t} \frac{1}{m_t} \sum_{j}   K_0(x_{t''j''
},x_{tj})_{\cdot t} (f_t(x_{tj})-\beta_t), 
\end{talign*}
$\forall t'' \in [T], j'' \in [m_T]$.

Furthermore, if $K_0$ is strictly positive definite and the points $x_{tj}$ are distinct, then the system is strictly convex and $\theta^*$ is unique.
\end{theorem}

See \Cref{appendix:proof optimal function} for the proof. \Cref{thm:closed_form_minimiser} assumes that $\beta$ is known and fixed, which may not be the case in practice.
However,  given a fixed $g_\theta$ the objective \Cref{eq:Lvv_m} is a quadratic in $\beta$ with
the optimal value $\beta^*_t = \frac{1}{m_t} \sum_{j=1}^{m_t}  f_t(x_{tj})- (g_{\theta}(x_{tj}))_t$. This naturally leads to the use of block coordinate descent approaches which iterate between optimising $\beta$ and $\theta$. This could be either directly implemented using the closed-form solutions, or through the use of numerical optimisers such as the stochastic optimisation approaches we will now present. The next paragraph highlights how this can be implemented for \emph{special case I and II}. We remark that the use of stochastic optimisation tools will be essential in most applications due to the size of the linear systems leading exact solutions to being intractable in practice.

\begin{algorithm}[tb]
\caption{Block-coordinate descent for \gls{vvCV}s  with unknown task relationship}
\label{alg:vvCV_block} 
\begin{algorithmic}
   \STATE {\bfseries Input:} $\D, \tilde{m}, L, \lambda, \beta^{(0)}, \theta^{(0)}$ and $B^{(0)}$
   \FOR{iteration $l=1$ {\bfseries to} $L$}
   \STATE Select a mini-batch $\D_{\tilde{m}}$ of size $\tilde{m} = (\tilde{m}_1, \ldots, \tilde{m}_T)^\top$
   \STATE  $\big(\theta^{(l)},\beta^{(l)} \big) \leftarrow \textsc{Update}_{\theta,\beta} (\theta^{(l-1)}, \beta^{(l-1)}, B^{(l-1)}; \D_{\tilde{m}}\big)$
    \STATE  $B^{(l)} \leftarrow \textsc{Update}_B \big( \theta^{(l)},\beta^{(l)}, B^{(l-1)};\D_{\tilde{m}}\big)$.
   \ENDFOR
   \STATE {\bfseries Return:} $\theta^{(L)}$, $\beta^{(L)}$ and $B^{(L)}$
\end{algorithmic}
\end{algorithm}

\paragraph{Unknown Task Relationship} 
\label{sec:learn_B_jointly_case1}

We now extend our approach to account for simultaneously estimating $\beta$, $\theta$, but also the matrix $B$. In this setting, we will use the same objective as in  \Cref{eq:Lvv_m}, but penalised by the norm of $B$:
\begin{talign}
\label{eq:vv_CV_learnB_1D1P_raw_obj}
    \bar{L}^{\text{vv}}_m(\theta, \beta, B)  = \Jvv_m(\theta,\beta,B) + \lambda\|g_\theta\|^2+ \|B\|^2,
\end{talign}
We now use $\Jvv_m(\theta,\beta,B)$ to denote $\Jvv_m(\theta,\beta)$ in order to emphasise the dependence on $B$. 
A second regularisation parameter is unnecessary as this would be equivalent to rescaling $k$. The objective in \Cref{eq:vv_CV_learnB_1D1P_raw_obj} is a natural extension to \Cref{eq:Lvv_m} and can be straightforwardly minimised through stochastic optimisation; see Algorithm \ref{alg:vvCV_block}. To ensure $B$ is strictly positive definite, we can take $B = L L^\top$, where $L$ is a lower triangular matrix with diagonal elements forced to be greater than zero via an exponential transformation. The pseudo-code in Algorithm \ref{alg:vvCV_block} presents this abstractly as a function $\textsc{Update}$. This is because different choices of \gls{vvCV}s might benefit from different updates. For example, pre-conditioners for the gradients could be used when readily available, or when these can be estimated from data. In \Cref{sec:experiments}, we will exclusively be using the Adam optimiser \citep{Kingma2015}, a first-order method with estimates of lower-order moments. Additionally, we study the convex case when $B$ is known and only $\beta$ and $\theta$ are required to be estimated in \Cref{appendix:fixedB}.

\paragraph{Computational Complexity}
Although \gls{vvCV}s can be beneficial from an accuracy viewpoint, they also incur a significant computational cost. Whether they should be used will therefore depend on the computational budget available. In particular, when evaluations of $f$ or $\nabla \log \pi$  are expensive, the higher cost of using \gls{vvCV}s may be negligible.  \Cref{tab:computation_complex} provides the computational complexity of all approaches considered in the paper. We emphasise the impact of $T$, the number of tasks. In the case of existing kernel-based \gls{CV}s, the dependence is $\O(T)$. In contrast, the computational cost of \gls{vvCV}s is between $\O(T^4)$ and $\O(T^6)$. \Cref{tab:computation_complex} also highlights the difference in computational complexity between obtaining closed form solutions of $\theta$ and $\beta$  by solving the linear system of equations in \Cref{thm:closed_form_minimiser}, and using the stochastic-optimisation approaches in \Cref{appendix:fixedB} and 
\Cref{sec:learn_B_jointly_case1}. Here, the difference is mainly in terms of powers of $T$, $m_t$ and $\tilde{m}_t$. As we can see the exact solutions usually come with an $\O(m_t^3)$ cost, whereas the stochastic optimisation approach is associated with a $\O(m_t \tilde{m}_t L)$ cost. In this case, whenever $\tilde{m}_t L$ is small relative to $m_t$, this will lead to computational gains.

In all applications considered, both $m_t$ and $T$ were small so the overall cost is controlled. However, this computational complexity can be further significantly reduced in special cases. When using a kernel corresponding to a finite-dimensional RKHS (e.g. a polynomial kernel), the scaling becomes linear in $m_t$, but is $\O(q^3)$ instead of $\O(m_t^3)$, where $q \ll m_t$ is the dimensionality of the RKHS. Alternatively, for certain choices of point sets and kernels, it is possible to reduce the computational complexity to $\O(m_t \log m_t)$ instead of $\O(m_t^3)$ by using scalable kernel methods such as fast Fourier features or inducing points. When the integrands are evaluated at the same set of points and the separable kernel is used, the computational cost in $T$ also becomes $\O(T^2)$ instead of $\O(T^6)$, once again significantly reducing the computational complexity.

\begin{table}[t]
\caption{Computational complexity of kernel-based \gls{CV}s and \gls{vvCV}s as a function of $d, m, \tilde{m}, L$ and $T$. We assume that $m_t$ is the same $ \forall t \in [T]$ up to additive or multiplicative constants (and similarly for all $\tilde{m}_t$ with $t \in [T]$). The cost of stochastic optimisation algorithms is assumed to only scale with the cost of stochastic estimates of the gradient of $\Jvv$. }
\label{tab:computation_complex}
\begin{center}
\begin{sc}
\scalebox{0.85}{
\begin{tabular}{lcccr}
\toprule
\textit{Method} & CV & vv-CV \\
\hline 
\textit{Exact solution} & $\O((d m_t^2 + m_t^3)T) $ & $\O(d m_t^2 T^4 + m_t^3 T^6)$  \\
\hline
\textit{Stochastic optim.} & $\O(d \tilde{m}_t m_t L T)$ & $\O(d\tilde{m}_t m_t L T^4)$ \\
\bottomrule
\end{tabular}
}
\end{sc}
\end{center}
\vskip -0.15in
\end{table}

\section{Experimental Results}
\label{sec:experiments}

We now illustrate our method on a range of problems including multi-fidelity models, computation of the model evidence for dynamical systems through thermodynamic integration and Bayesian inference for the abundance of preys using a Lotka-Volterra system. See \Cref{appendix:additional_experiments} for additional experiments including illustrations of matrix-valued Stein kernels $K_0$ in \Cref{appendix:illustration_mvKernel} and a synthetic example when the Stein kernel matches the smoothness of integrands in \Cref{appendix:south_experiments}. Since we are interested in gains obtained from the \gls{CV}s, we fix $n=0$ which means we are using all the data to construct \gls{vvCV}s. The code to reproduce our results is available at: \url{https://github.com/jz-fun/Vector-valued-Control-Variates-Code}.

\subsection{Multidelity Modelling in the Physical Sciences}\label{sec:multifidelity_modelling}

Many problems in the engineering and physical sciences can be tackled with multiple models of a single system of interest. These models are often associated with varying computational costs and levels of accuracy, and their combination to solve a task is usually called multi-fidelity modelling; see \cite{Peherstorfer2018} for a review. 
We will consider a high-fidelity model $f_H$ and a low-fidelity model $f_L$, and will attempt to estimate the integral of $f_H$ with our \gls{vvCV}s and using function evaluations from both the high- and low-fidelity models. For clarity, we will now denote the function $f=(f_L,f_H)$ and the vector-probability distribution $\Pi = (\Pi_L,\Pi_H)$. We note that this is a special case of the problem considered in our paper since we use evaluations of multiple functions but are only interested in $\Pi_H[f_H]$ (whereas $\Pi_L[f_L]$ is not of interest).

\begin{figure}[t!]
\begin{center}
\centering
\includegraphics[width=\columnwidth,trim=0cm 0cm 25.8cm 0cm, clip=true]{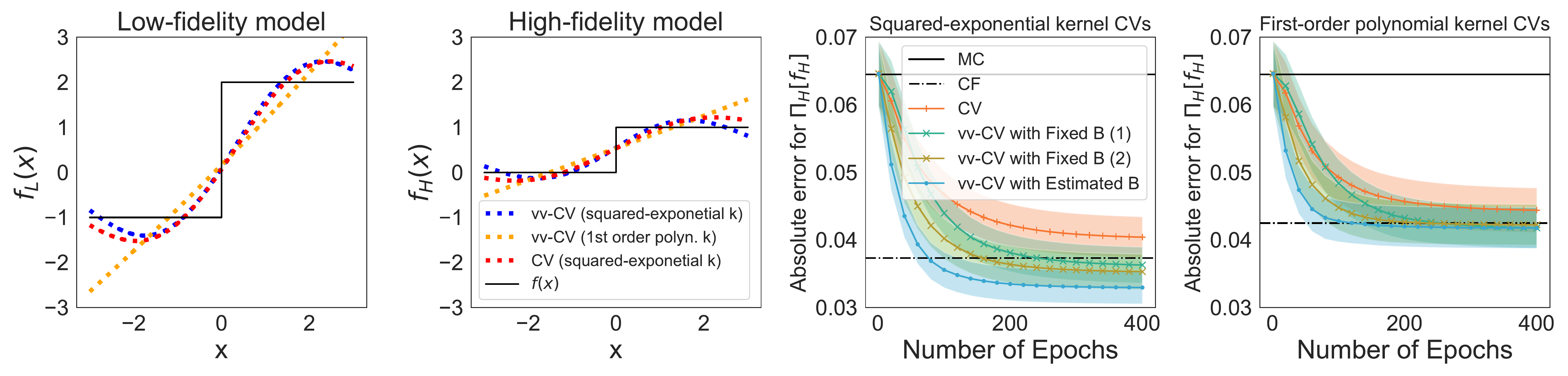}
\includegraphics[width=\columnwidth,trim=25cm 0cm 0cm 0cm, clip=true]{IMAGES/step_function_plot.pdf}
\caption{ \emph{Numerical integration of univariate discontinuous multifidelity model.} 
\emph{Upper}: fitted CVs for both functions. \emph{Lower Left:} performance of \gls{CV}s based on a squared-exponential kernel for the high-fidelity function as a number of epochs of the optimisation algorithm. The lines provide the mean over $100$ repetitions of the experiment, whereas the shaded areas provide one standard deviation above and below the mean. 
\emph{Lower Right:} same experiment for a polynomial kernel for the high-fidelity function.  
} 
\label{fig:multifidelity_univariate}
\end{center}
\vskip -0.3in
\end{figure}

\vspace{-0.1in}
\paragraph{Univariate Step Function} The first example considered is a toy problem from the multi-fidelity literature \citep{ICML2018_BQforMultipleRelatedIntegrals}. The low-fidelity function is $f_{\text{L}}(x) = 2~\text{if}~x \ge 0$ and $-1~\text{otherwise}$. The high-fidelity function is $f_{\text{H}}(x) = 1~\text{if}~x \ge 0$ and $0~\text{otherwise}$. In this example, $K_0$ is smoother than $f_1$ and $f_2$, which are both discontinuous.  The integral is over the real line and taken against $\Pi = \mathcal{N}(0,1)$, and we fix the sample sizes to $m = (m_{\text{L}},m_{\text{H}})=(40,40)$. 

Results with a squared-exponential and  1$^{\text{st}}$ order polynomial kernel $k$ can be found in \Cref{fig:multifidelity_univariate}.  The upper plots clearly show that the approximations are not of very high-quality, but the lower plots show that all \gls{CV}s can still lead to an order of two gain in accuracy over \gls{MC} methods. We also observe that \gls{vvCV}s can lead to further gains over existing \gls{CV}s by leveraging evaluations of $f_L$. For both kernels, we provide three different versions of the \gls{vvCV}s with separable structure to highlight the impact of the matrix $B$. The first two cases use Algorithm \ref{alg:SGD_vvCV} with a fix value of $B$. In the first instance, $B_{11} =  B_{22} = 0.1, B_{12}=B_{21} = 0.01$, whereas in the second instance $B_{11} = B_{22} = 0.5, B_{12} = B_{21} = 0.01$. The third case is based on estimating $B$ through Algorithm \ref{alg:vvCV_block}. Clearly, $B$ can have a significant impact on the performance of the \gls{vvCV}, and estimating a good value from data can provide further gains. The choice of $k$ is also significant: all \gls{CV}s based on the squared-exponential kernel significantly outperform the \gls{CV}s based on a 1$^{\text{st}}$ order polynomial kernel. It is also found that even when the model is mis-specified, the proposed method still perform better than standard scalar-valued \gls{CV}s.

\paragraph{Modelling of Waterflow through a Borehole} A more complex example often used to assess multifidelity methods is the following model of water flow passing through a borehole \citep{Xiong2013,Kandasamy2016,Park2017}. Both $f_L$ and $f_H$ have $d=8$ inputs representing a range of parameters influencing the geometry and hydraulic conductivity of the borehole, as well as transmissivity of the aquifer. Prior distributions have been elicited from scientists over input parameters to account for uncertainties about their exact value. See \Cref{appendix:multifidelity_modelling} for the details of each input and the multi-fidelity models. One quantity of interest here is the expected water flow under these  distributions, and we hence have $\Pi_L=\Pi_H$. 

\vspace{-0.1in}
\begin{table}[h]
\caption{\emph{Expected values of the flow of water through a borehole}. The numbers provided give the mean absolute integration error for $100$ repetition of the task of estimating $\Pi_H[f_H]$, and the numbers in brackets provide the sample standard deviation. To provide the absolute error, the true value ($72.8904$) is estimated by a \gls{MC} estimator with $5\times 10^5$ samples.}
    \label{tab:borehole_balance_case}
\vskip 0.15in
\begin{center}
\begin{sc}
\scalebox{0.75}{
\begin{tabular}{c|| cc cc }
\toprule
        $m$ &  \gls{vvCV}- Est. B  &  \gls{vvCV}-Fix. B & CF & MC\\ \hline \hline
         $10$ & 3.72 (0.27)  & \textbf{1.94} (0.15) & 2.24 (0.16) & 6.42 (0.44)  \\ \hline
         $20$ & \textbf{1.29} (0.10)  & 1.35 (0.10) & 1.96 (0.10)  & 4.31 (0.31)  \\ \hline
         $50$ & \textbf{1.04} (0.06)  & 1.77 (0.12) & 1.76 (0.07)  & 2.63 (0.17)  \\ \hline
         $100$ & \textbf{1.07} (0.06)  & 1.65 (0.14) & 1.71 (0.05)  & 1.83 (0.15)  \\ \hline
         $150$ & \textbf{0.85} (0.05)  & 1.30 (0.09) & 1.67 (0.04)  & 1.42 (0.10)  \\
    \bottomrule
\end{tabular}
}
\end{sc}
\end{center}
\vskip -0.1in
\end{table}

\noindent Results of our simulation study are presented in \Cref{tab:borehole_balance_case}. We compare a standard \gls{MC} estimator with a kernel-based \gls{CV} fitted with a closed form solution (denoted \gls{CF}) and two kernel-based \gls{vvCV}s corresponding to special case II in \Cref{sec:methodology}. The first with $B_{11} = B_{22} = 5 \times 10^{-4}$ and $B_{12} = B_{21} = 5 \times 10^{-5}$, and the second with $B$ estimated using Algorithm \ref{alg:vvCV_block}. The kernel used is a tensor product of squared-exponential kernels with a separate lengthscale for each dimension. Clearly, \gls{vvCV}s significantly outperform \gls{MC} in the large majority of cases, and estimating $B$ can lead to significant gains over using a fixed $B$. The worst performance for \gls{vvCV}s with estimated $B$ is when values of $m$ are the lowest. This is because $m$ is not large enough to learn a good $B$.  See \Cref{appendix:multifidelity_modelling} for further details.


\subsection{Model Evidence for Dynamical Systems}

We now consider Bayesian inference for non-linear differential equations such as dynamical systems, which can be particularly challenging due to the need to compute the model evidence. This is usually a computationally expensive task since sampling from the posterior repeatedly requires the use of a numerical solver for differential equations which needs to be used at a fine resolution. 

\begin{figure}[ht]
\vskip 0.in
\begin{center}
\centerline{\includegraphics[width=1\columnwidth, trim={0.1cm 0.cm 0.1cm 0.cm},clip]{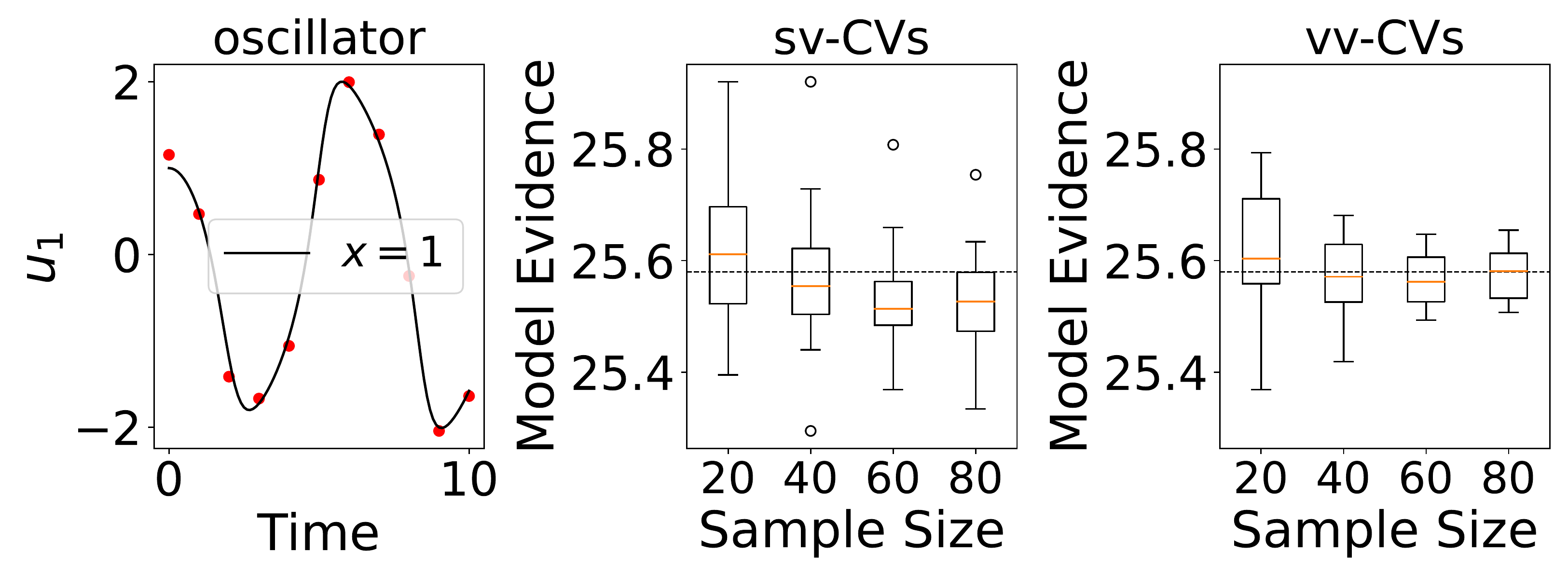}}
\caption{\emph{Model evidence computation through thermodynamic integration.} \emph{Left:} Illustration of the van der Poll oscillator model (black line) and corresponding observations (red dots). \emph{Center:} Estimates of the model evidence as a function of the number of posterior samples for kernel-based \gls{CV}s. The box-plots were created by repeating the experiment $20$ times and the black line gives an estimate of the truth obtained from \cite{oates2017_CF_for_MonteCarloIntegration} ($25.58$). \emph{Right:} Same experiment but with kernel-based \gls{vvCV}s. }
    \label{fig:thermodynamic_integration}
\end{center}
\vskip -0.3in
\end{figure}

In \cite{Calderhead2009}, the authors propose to use \gls{TI} \citep{friel2008_marginal_llk_by_power_posterior} to tackle this problem, and \cite{oates2014_controlled_themo_integration,oates2017_CF_for_MonteCarloIntegration} later showed that \gls{CV}s can lead to significant gains in accuracy in this context.  \gls{TI} introduces a path from the prior $p(x)$ to the posterior $p(x|y)$, where $y$ and $x$ represent the observations and the unknown parameters respectively. This is accomplished by the power posterior $p(x|y, t) \propto p(y|x)^t p(x)$, where $t \in [0,1]$ is called the inverse temperature. When $t = 0$, $p(x|y, t) = p(x)$, whereas when $t=1$, $p(x|y, t) = p(x|y)$. The standard \gls{TI} formula for the model evidence has a simple form  which can be approximated using second-order quadrature over a discretised temperature ladder $0 = t_1 \le \cdots \le t_w =1$ \cite{friel2014_improving_power_posterior}. It takes the following form
\begin{talign*} 
  &\log p(y) = \int_0^1 \left[\int_\mathcal{X} \log p(y|x) p(x|y, t) d \theta\right] dt \\
  &\qquad \approx \sum_{i=1}^{w} \frac{t_{i+1} - t_i}{2} ( \mu_{i+1} + \mu_i) - \frac{(t_{i+1} - t_i)^2}{12} (v_{i+1} - v_i) ,
\end{talign*}
where $\mu_i$ is the mean and $v_i$ the variance of the integrand $f(x) := \log p(y|x)$ with respect to $\pi_i(x) := p(x|y, t_i)$. To estimate $\{\mu_i,v_i\}_{i=1}^{w}$, we need to sample from all power posteriors on the ladder, then use a \gls{MCMC} estimator which can be enhanced through \gls{CV}s. This gives $T=2w$ integrals which are \emph{related}: $w$ integrals to compute means and $w$ integrals to compute variances, each against different power posteriors. As we will see, this structure will allow \gls{vvCV}s to provide significant gains in accuracy.

Our experiments will focus on the \emph{van der Poll oscillator}, which is an oscillator $u:\R_+ \times \mathcal{X} \rightarrow \R$ (where $x \in \R$) given by the solution of $d^2 u/ ds^2 - x (1- u^2) du/ds + u = 0$, where $s$ represents the time index. For this experiment, we will follow the exact setup of \cite{oates2017_CF_for_MonteCarloIntegration} and  transform the equation into a system of first order equations: 
$\nicefrac{d u_1}{ds} = u_2, \nicefrac{d u_2}{ds} = x(1-u_1^2)u_2-u_1$, which can be tackled with ODE solvers. Our data will consist of noisy observation of $u_1$ (the first component of that system) given by $y(s) = \mathcal{N}(u_1(s; x), \sigma^2)$ with $\sigma=0.1$ at each point $s \in \{0,1,\ldots,10\}$; see the left-most plot of \Cref{fig:thermodynamic_integration} for an illustration. We will take a ladder of size $w = 31$ with $t_i = ((i-1)/30)^5$ for $i \in \{1,\ldots,31\}$. This gives a total of $T=62$ integrals will need to be computed simultaneously, which is likely to be too computationally expensive for \gls{vvCV}s in their full generality. As a result, we chose $B$ to be a block diagonal matrix which puts integrands in groups of $4$ means or $4$ variances (except one group of 3 for mean and variance). To sample from the power-posteriors, we use population \gls{MC} with the manifold Metropolis-adjusted Langevin algorithm \citep{Girolami2011}. Due to the high computational cost of using ODE solvers, our number of samples will be limited to less than $100$ per integrand and this number will be the same for each integration task.

Our results are presented in \Cref{fig:thermodynamic_integration}. 
The kernel parameters were taken to be identical to those in \cite{oates2017_CF_for_MonteCarloIntegration}.
As observed in the centre plot, kernel-based \gls{CV}s provide relatively accurate estimates of the model evidence. As the sample size increases, we notice less variability in these estimates, but the central $50\%$ of the runs are contained in an interval which excludes the true value. In comparison, the right-most plot shows that kernel-based \gls{vvCV}s can provide significant further reduction in variance. The distribution of estimates is also much more concentrated and centered around the true value.

\subsection{Bayesian Inference of Lotka-Volterra System}

We now consider another model: the Lotka-Volterra system \citep{lotka1925elements, volterra1926variazioni, lotka1927fluctuations} of ordinary differential equations.  This system is given by: $\nicefrac{d v_1(s)}{ds} = \alpha v_1(s)-\beta v_1(s) v_2(s), \nicefrac{d v_2(s)}{ds} = \delta v_1(s) v_2(s)-\gamma v_2(s)$.
Here, $s \in [0,S]$ for some $S \in \R_+$ denotes the time, and $v_1(s)$ and $v_2(s)$ are the numbers of preys and predators, respectively. The system has initial conditions $v_1(0)$ and $v_2(0)$.
We have access to noisy observations of $v = (v_1,v_2)$ at points $s_1,\ldots,s_m \in [0,S]$ denoted  $y_{1j},y_{2j}$ and which are both observed with log-normal noise with standard deviation $\sigma_{y_1}$ and $\sigma_{y_2}$ respectively for all $j \in \{1,\ldots,m\}$, given some unknown parameter value $x^* =(\alpha^*,\beta^*, \delta^*, \gamma^*, v_1(0)^*, v_2(0)^*, \sigma_{y_1}^*, \sigma_{y_2}^*)^\top$. In practice, we reparameterise $x$ such that the model parameters are defined in $\R^8$; see \Cref{appendix:Lotka} for details. Given these observations, we can construct a posterior $\Pi$ on the value of $x^*$. We will then be interested in computing posterior expectations of $v_1$ at a set of time points $s'_1,\ldots,s'_T$, and hence have $T$ integrands of the form $f_t(x) = v_1(s_t';x)$ where $x$ highlights the dependence on the parameter $x$. \gls{CV}s were previously considered for individual tasks in this context by \citet{Si2020}. However, these $T$ integrands are related when $s'_1,\ldots,s'_T$ are close to each other.

 We use the dataset of snowshoe hares (preys) and Canadian lynxes (predators) from \citet{hewitt1921conservation}, and implement Bayesian inference on model parameters $x$ by using no-U-turn sampler (NUTS) in Stan \citep{carpenter2017stan}. For sv-CVs, we estimate each individual $\Pi[f_t]$ separately; while for \gls{vvCV}s, we estimate a collection these tasks $\Pi[f] := (\Pi[f_1], \ldots, \Pi[f_T])^\top$ jointly. See \Cref{appendix:Lotka} for experimental details. 
In \Cref{tab:Lotka2}, we compare kernel \gls{vvCV}s (special case II) with standard MCMC estimators and CF estimators. We consider two cases of \gls{vvCV}s: the first is the case when $B$ is with $B_{tt} = 5 \times 10^{-4}$  for all $t$ and $B_{tt'} = 5 \times 10^{-5}$ for $\forall t \neq t'$, and the second with $B$ estimated using Algorithm 2. The kernel used is a tensor product of squared-exponential kernels with a separate lengthscale for each dimension. We increase the number of tasks from $T=2$ to $T=10$. Once again, \gls{vvCV}s significantly outperforms MCMC, especially for large $T$, and estimating $B$ provides further gains over using a fixed $B$.

\begin{table}[t]
\caption{\emph{Posterior Expected Abundance of Preys.}  The numbers provided give the sum of the mean absolute integration error for $10$ repetition of each task. To provide the absolute error, the true values of the associated expectations are estimated by \gls{MCMC} estimators with $8\times 10^5$ posterior samples.}
    \label{tab:Lotka2}
\vskip -0.1in
\begin{center}
\begin{small}
\begin{sc}
\scalebox{0.82}{
\begin{tabular}{c c|| cc  cc  cc  cc}
\toprule
$T$  & $m$ &  \gls{vvCV}- Est. B  & \gls{vvCV}-Fix. B   & CF & MCMC \\ \hline \hline
        $2$  & $500$ &  0.462 & \textbf{0.404}  & 0.666 & 0.568  \\        \hline
        $5$ & $500$  & \textbf{0.393}  &  0.419  & 0.521 &  0.987    \\ \hline
        $10$  & $500$  & \textbf{0.938}   & 1.031 & 2.540  &   2.663\\
\bottomrule
\end{tabular}
}
\end{sc}
\end{small}
\end{center}
\vskip -0.25in
\end{table}


\section{Conclusion}

This paper considered variance reduction techniques that share information across related integration problems. The proposed solution, \emph{vector-valued control variates}, was shown to lead to significant variance reduction for problems in multi-fidelity modelling and Bayesian computation. Our approach is, to the best of our knowledge, the first algorithm able to perform multi-task learning for numerical integration by using only evaluation of the score functions of the corresponding target distributions. It is also the first algorithm which can simultaneously learn the relationship between integrands and provide estimates of the corresponding integrals without the requirement of a tractable kernel mean as in \citep{ICML2018_BQforMultipleRelatedIntegrals}.

On an algorithmic level, further work will be needed to make the method more computationally practical and efficient. One particular line of research which could be considered is how special cases of our matrix-valued Stein kernels in \Cref{thm:K0_kernel_firstorder} can be selected to reduce the computational cost whilst still producing a rich class of \gls{vvCV}s. Since a preprint version of this paper appeared online, it has been shown that approaches based on meta-learning could be competitive for very large $T$; see \citet{Sun2023}.

On a theoretical level, we could also look at the question of when transferring information across tasks will lead to sufficient gains in accuracy to warrant the additional computational cost.  In addition, it could be of interest to understand the negative transfer problems which can arise when the integrands are not in \gls{RKHS} $\H_{K_0}$ or when the sample size is too small to estimate $B$ well.

\section*{Acknowledgements}
The authors would like to thank Chris J. Oates for helpful discussions and for sharing some of his code for the thermodynamic integration example. 
ZS was supported under the EPSRC grant [EP/R513143/1]. 
AB was supported by the Department of Engineering at the University
of Cambridge, and this material is based upon work supported by, or in part by, the U.S. Army Research Laboratory and the U. S. Army Research Office, and by the U.K. Ministry of Defence and under the EPSRC grant [EP/R018413/2].
FXB was supported by the Lloyd’s Register Foundation Programme on Data-Centric Engineering and The Alan Turing Institute under the EPSRC grant [EP/N510129/1], and through an Amazon Research Award on
“Transfer Learning for Numerical Integration in Expensive Machine Learning Systems”.


\bibliography{example_paper}
\bibliographystyle{icml2023}

\newpage
\appendix
\onecolumn

{
\begin{center}
\LARGE
    \vspace{5mm}
    \textbf{Appendix}
    \vspace{5mm}
\end{center}
}

We now complement the main text with additional details. Firstly, we provide additional methodology in \Cref{appendix:additional_method}, including alternative objectives based on the variance of \gls{MCMC}, stochastic optimisation algorithm for \gls{vvCV}s when the task relationship $B$ is known (under special case I and II) and the calculation of the computational complexity that is presented in \Cref{tab:computation_complex}. Secondly, in \Cref{appendix:proofs}, we provide the proofs of all the theoretical results in the main text. Then, in \Cref{appendix:alternative_constructions}, we present alternative constructions of \gls{vvCV}s. These constructions include a more general form of kernel-based \gls{vvCV}s as well as some polynomial-based \gls{CV}s. Finally, in \Cref{appendix:implementation}
 and \Cref{appendix:additional_experiments}, we provide additional details on our implementation of \gls{vvCV}s and provide additional numerical experiments.


\section{Additional Methodology}
\label{appendix:additional_method}
In this first appendix, we will present additional methodology for \Cref{sec:methodology}. We briefly discuss an alternative objective based on the variance of the MCMC case in \Cref{appendix:alter_obj_MCMC}. The selection of \gls{vvCV}s under \emph{special case I and II} when $B$ is known and the corresponding algorithm is presented in \Cref{appendix:fixedB}.

\subsection{Alternative Objectives based on the Variance of MCMC}
\label{appendix:alter_obj_MCMC}

We present an alternative objective based on the variance of MCMC here, which is given by the following remark.
\begin{remark}\label{remark:MCMC_objective}
Given $\{X_i\}_{i=1}^{\infty}$ an ergodic Markov chain with invariant distribution $\Pi$, the variance of the \gls{MCMC} \gls{CLT} is proportional to 
\begin{talign*}
    J(\theta) + 2\sum_{i=1}^{\infty} \text{Cov}_{\Pi}[ f(X_1)-g_{\theta}(X_1), f(X_i)-g_{\theta}(X_i)] 
\end{talign*}
where the second term accounts for the correlation in the Markov chain. Given some realisation of the Markov chain  $\{x_j\}_{j=1}^m$ and the corresponding evaluations $\{f(x_j)\}_{j=1}^m$, this can be approximated as:
\begin{talign*}
       J_m(\theta,\beta) +2\sum_{s=1}^{m-1} \frac{1}{m}\sum_{i=1}^{m-s}(f(x_i) - g_{\theta}(x_i) -\MC[f-g_{\theta}])   (f(x_{i+s}- g_{\theta}(x_{i+s}) - \MC[f-g_{\theta}])).
\end{talign*}
\end{remark}

\subsection{Stochastic Optimisation with a Known Task Relationship} \label{appendix:fixedB}
In this section, we discuss the selection of \gls{vvCV}s under \emph{special case I and II} when $B$ is known and give the corresponding algorithm in Algorithm~\ref{alg:SGD_vvCV}.

In this case, we will assume that $B \in S_+^T$ is known. The proposed algorithm is a stochastic optimisation algorithm which we run for $L$ time steps; pseudo-code is provided in Algorithm \ref{alg:SGD_vvCV}. We propose to initialise the algorithm at $\theta^{(0)}=(0,\ldots,0) \in \R^p$, since this is equivalent to having $g_\theta(x) = 0$ (i.e. having no CV) for both the kernel- and polynomial-based \gls{vvCV}s. We also suggest initialising the parameter $\beta \in \R^T$ with any estimate of $\Pi[f]$. This is a natural initialisation since we expect $\beta_t$ to equal $\Pi[f_t]$ for all $t \in [T]$ when $m_1,\ldots,m_T \rightarrow \infty$. For example, when the data consists of \gls{IID} realisations from $\Pi_1,\ldots, \Pi_T$, a natural initialisation point is $\beta^{(0)} = (\MC_1[f_1],\ldots,\MC_T[f_T])^\top$. 

For each iteration, we take mini-batches of size $\tilde{m} \in \N^T$ where $|\tilde{m}| \leq |m|$. Note here that $\tilde{m}=(\tilde{m}_1,\ldots,\tilde{m}_T)^\top$ is a multi-index giving the size of the mini-batch for each of the $T$ datasets. This formulation allows for the use of different datasets across integrands, but also different mini-batch sizes for each task (which may be useful if the datasets are of different size for each integrand). An epoch consists of having gone through all data points for all $T$ tasks, and we randomly shuffle the indices for mini-batches after each epoch. As default, we propose to take $\tilde{m}_t  \propto m_t /(\sum_{t=1}^T m_t)$ for all $t \in [T]$. This choice guarantees that the number of samples for each integrand in the mini-batches is proportional to the proportion of samples for that integrand in the full dataset.

Once a mini-batch has been selected, we update our current estimate of the parameters $\theta$ and $\beta$ using steps based on the gradient of our loss function: $\nabla_{(\theta,\beta)} \Lvv_{\tilde{m}}(\theta,\beta)$. The pseudo-code in Algorithm \ref{alg:SGD_vvCV} presents this abstractly as a function $\textsc{Update}_{\theta,\beta} (\theta, \beta, B; \D )$, which takes in the current estimates of the parameters, the value of $B$ and the dataset (or minibatch) used for the update; this is because different choices of \gls{vvCV}s might benefit from different updates. For example, pre-conditioners for the gradients could be used when readily available, or when these can be estimated from data. In \Cref{sec:experiments}, we will exclusively be using the Adam optimiser \citep{Kingma2015}, a first order method with estimates of lower-order moments. 

For the penalisation term, it would be natural to take the \gls{RKHS} norm $\|g_\theta\| = \|g_\theta\|_{\H_{K_0}}$ since this would lead to the objective used in \Cref{thm:closed_form_minimiser}.
However, this can be impractical from a computational viewpoint since this norm depends on kernel evaluations for all of the training points.
For this reason, we follow the recommendation of \cite{Si2020} and use instead the Euclidean norm: $\|g_\theta\| = \|\theta\|_2$.
This still leads to a convex objective since the objective remains quadratic in $\theta$.

Algorithm \ref{alg:SGD_vvCV} is a natural approach to minimising our objective since our kernel-based \gls{vvCV}s are linear in $\theta$ and $\Lvv_m$ is convex in $(\theta,\beta)$. Many stochastic optimisation methods, such as stochastic gradient descent, will hence converge to a global minimum under regularity conditions \citep{Bottou2018}. However, note that Algorithm \ref{alg:SGD_vvCV} naturally applies to other \gls{vvCV}s whether linear or not.

\begin{algorithm}[tb]
\caption{Stochastic optimisation for \gls{vvCV}s with known task relationship}
\label{alg:SGD_vvCV}
\begin{algorithmic}
   \STATE {\bfseries Input:} $\mathcal{D}$, $\tilde{m}$, $L$, $\lambda$, $\beta^{(0)}$ and $\theta^{(0)}$.
   \FOR{iterations $l$ from $1$ {\bfseries to} $L$}
   \STATE  Select a mini-batch $\mathcal{D}_{\tilde{m}}$ of size $\tilde{m}$
   \STATE $\big(\theta^{(l)},\beta^{(l)} \big) \leftarrow \textsc{Update}_{\theta,\beta} (\theta^{(l-1)}, \beta^{(l-1)}, B; \mathcal{D}_{\tilde{m}}\big)$
   \ENDFOR
   \STATE {\bfseries Return:} $\theta^{(L)}$, $\beta^{(L)}$
\end{algorithmic}
\end{algorithm}

\subsection{Calculation of Computational Complexity}
\label{appendix:calculation_comp_cost}
In this section, we derive the computational complexity reported in \Cref{tab:computation_complex}. Suppose we have $m_t$ samples for each of $T$ tasks and similarly for all $\tilde{m}_t$ with $t\in [T]$.

\paragraph{Computational cost of CV:} 
\begin{itemize}
    \item  \textbf{Exact solution:} For each task, we need to compute $k_0(\cdot, \cdot)$ for all pairs, which results in a cost of $\O(d m_t^2)$ per task. To compute the exact solution of kernel-based control variates, it takes $\O(m_t^3)$ per task. So, in total, the computational cost of the exact solution of CV is $\O((d m_t^2 + m_t^3)T)$.

    \item   \textbf{Stochastic optimisation:} To use stochastic optimisation, suppose we use $L$ epochs. At each iteration of one epoch, we need to compute $k_0(\cdot, \cdot)$ for $\tilde{m}_t m_t$ pairs,  which costs $\O(d \tilde{m}_t m_t)$. 
    We need to do this for all $L$ iterations. This results in $\O(d \tilde{m}_t m_t L)$ per task. Hence, in total, the computational cost of stochastic optimisation of CV is $\O(d \tilde{m}_t m_t L T)$ for all $T$ tasks.

\end{itemize}

\paragraph{Computational cost of vv-CV:}  
\begin{itemize}
    \item \textbf{Exact solution:} There are $m_t T$ samples (i.e. $m_t^2 T^2$ pairs) in total and we are using mv-Stein kernels. Thus, the computational cost of computing $K_0(\cdot, \cdot)$ is $\O(d (m_t T)^2 T^2) = \O(d m_t^2 T^4 )$. To compute the exact solution of vv-CV, we can re-arrange the Gram tensor of all samples into a matrix of size $m_t T^2 \times m_t T^2$. 
    Hence, the computational cost of computing the exact solution of vv-CV is $\O((m_t T^2)^3) = \O(m_t^3 T^6)$. Thus, in total, the computational cost  of computing the exact solution of vv-CV is $\O(d m_t^2 T^4 + m_t^3 T^6)$.

    \item  \textbf{Stochastic optimisation:} At each iteration, one mini-batch of the stochastic optimisation of vv-CV has $\tilde{m_t} T$ samples. We need to compute $K_0(\cdot, \cdot)$ for these samples with all $m_t T$ samples. This leads to a cost of $\O(\tilde{m}_t T m_t T d T^2) = \O(d\tilde{m}_t m_t  T^4)$ per iteration. Note that we need to do this for all $L$ iterations. So, in total, the computational cost of stochastic optimisation of vv-CV is $\O(d\tilde{m}_t m_t L T^4)$.

\end{itemize}

\section{Proofs of the Main Theoretical Results} \label{appendix:proofs}

In this second appendix, we recall the proofs of the theoretical results in the main text. The derivation of the \gls{mvkernel} $K_0$ from \Cref{thm:K0_kernel_firstorder} can be found in \Cref{appendix:proof_K0_general}. The proof that kernel-based \gls{vvCV}s are square-integrable (\Cref{thm:squared inegrable RKHS}) is in \Cref{app:proof square integrable RKHS}. Finally, the proof of the existence of the optimal parameters as the solution to a linear system (\Cref{thm:closed_form_minimiser}) is given in \Cref{appendix:proof optimal function}.


\subsection{Proof of \Cref{thm:K0_kernel_firstorder}}
\label{appendix:proof_K0_general}
\begin{proof}
We will show $K_0$ is a kernel and derive its matrix components by constructing an appropriate  feature map.
The first order Stein operator maps matrix-valued functions $u = (u_1 , u_2 \ldots , u_T):\R^d \to \R^{d \times T}$ to the vv-function $\Svv[u]: \R^d \to \R^T$ given by
\begin{talign*}
\Svv[u] & = \left(
\L_{\Pi_1}\left[u_1\right], \ldots, 
\L_{\Pi_T}\left[u_T\right]\right)^\top  \\
\text{ where } \quad
\L_{\Pi_t}[u_t](x) & = \nabla_x \cdot u_t(x) + \nabla_x \log \pi_t(x) \cdot u_t(x) \quad \forall t \in [T].
\end{talign*}
Since $K \in C^{1,1}(\X \times \X)$,  we can use Corollary 4.36 of \cite{steinwart2008_SVM} to conclude that $\H_K$ is a vector-valued RKHS of continuously differentiable functions 
from $\R^d$ to $\R^{T}$, hence
the tensor product $\H_K^d$ consists of suitable functions
 $u \in \H_K^d$, with components $u^i = (u_1^i,\ldots, u_T^i) \in \H_K$ for $ i \in [d]$. Now since the RKHS consists of differentiable functions, we have by Lemma C.8 in \citet{barp2022targeted}:
\begin{talign*}
\metric{\partial^{j}_x K(\cdot, x) e_t }{u^i}_{\H_K} = e_t \cdot \partial^j u^i(x)= \partial^j u_t^i(x) \defn  \frac{\partial u_t^i}{\partial  x^j}(x) \quad \forall t \in [T]
\end{talign*}
where $e_t \in \R^T$ is a vector of zeros with value $1$ in the $t^{\text{th}}$ component. Then, writing $K^{e_t}_x \defn K(\cdot,x)e_t$, 
\begin{talign*}
\Svv[u](x) &=
\sum_{t=1}^T \sum_{r=1}^d (\partial^r_x u_t^r(x) + \partial^r_x \log \pi_t \, u_t^r(x)) e_t \\
& = \sum_{t=1}^T \sum_{r=1}^d \metric{\partial^r_x K_x^{e_t}+l_{tr}(x) K_x^{e_t}}{u^r}_{\H_K} e_t
\\
&=
 \sum_{t=1}^T\metric{\partial^\bullet_x K_x^{e_t}+l_{t\bullet}(x) K_x^{e_t}}{u}_{\H_K^d} e_t,
\end{talign*}
where $l_{tr}(x) = \partial^r_x \log \pi_t(x)$,
and
$\partial_x^\bullet K_x^{e_t}$ and  $l_{t\bullet}(x)$ denote respectively the tuples $(\partial_x^1 K_x^{e_t},\ldots, \partial_x^d K_x^{e_t}) \in \H_K^d$
and 
$(l_{t1}(x), \ldots, l_{td}(x)) \in \R^d$.

We have thus obtained  a feature map, i.e., a map 
$\gamma: \X \to \mathcal B(\H_K^d,\R^T)$, where $\mathcal B(\H_K^d,\R^T)$ denotes the space of bounded linear maps from $\H_K^d$ to $\R^T$, via the relation 
 \begin{talign*}
 \gamma(x) [u] =
\Svv[u](x),
 \end{talign*}
with adjoint $\gamma(x)^* =  \sum_{t=1}^T( \partial^\bullet_x K_x^{e_t}+l_{t \bullet}(x) K_x^{e_t})e_t$. Recall the adjoint map $\gamma(y)^* \in \mathcal B(\R^T,\H_K^d)$ to $\gamma(y)$,  is defined for any $a \in \R^T, u \in \H_K^d$ by the relation
\begin{talign*}
\metric{ \gamma(y)^*[a]}{u}_{\H_K^d} = \gamma(y)[u] \cdot a.
\end{talign*}

In particular, by Proposition 1 of \citet{Carmeli2010} we have that
\begin{talign*}
K_0(x,y) \defn \gamma(x) \circ \gamma(y)^* \in \R^{T \times T}
\end{talign*}
will then be the kernel associated to the ``feature operator'' (that is, a surjective partial isometry whose image is $\H_{K_0}$) $\Svv:\H_K^d \to \H_{K_0}$.
Subbing in the expressions for the feature map and its adjoint derived above, and using  the equalities 
\begin{talign*}
\metric{\partial^{s}_x K(\cdot, x) e_t }{ \partial_y^r K(\cdot,y)e_{t'}}_{\H_K} 
=
e_t \cdot
 \partial^s_x\partial_y^r K(x,y)e_{t'} 
 = 
 (\partial^s_x\partial_y^r K(x,y))_{tt'} \qquad \forall t,t' \in [T]\\ 
 \text{and } \qquad \metric{\partial^{ss}_x K(\cdot, x) e_t }{ \partial_y^r K(\cdot,y)e_{t'}}_{\H_K}=
  (\partial^{ss}_x\partial_y^r K(x,y))_{tt'} \qquad \forall t,t' \in [T],
\end{talign*}
 which hold for any differentiable  matrix-valued kernel \citep{barp2022targeted},
 we obtain  the following expression for the components of $K_0$
\begin{talign*}
(K_0(x,y))_{tt'} & = \sum_{r=1}^d (\partial^r_x \partial^r_y K(x,y))_{tt'}+ l_{t'r}(y)(\partial^r_x K(x,y))_{tt'} \\
& \qquad  + l_{tr}(x) \partial_y^r (K(x,y))_{tt'} + l_{tr}(x)l_{t'r}(y) (K(x,y))_{tt'}.    
\end{talign*}
In particular for separable kernels (i.e. $K(x,y)=B k(x,y)$) we have 
\begin{talign*}
(K_0(x,y))_{tt'} = B_{tt'}\sum_{r=1}^d \partial^r_x \partial^r_y k(x,y)+ l_{t'r}(y)\partial^r_x k(x,y)+ l_{tr}(x) \partial_y^r k(x,y) + l_{tr}(x)l_{t'r}(y) k(x,y).    
\end{talign*}
\end{proof}


\subsection{Proof of \Cref{thm:squared inegrable RKHS}}
\label{app:proof square integrable RKHS}
\begin{proof}

Recall that if a scalar kernel 
$k$ satisfies $\int_{\X} k(x,x) \mathrm{d} \mu (x) < \infty$, then its RKHS consists of square $\mu$-integrable functions (for any finite measure $\mu$) \citep[Theorem 4.26]{steinwart2008_SVM}. 

If $g \in \H_{K_0}$ then $g_t$ belongs to the RKHS with scalar-valued kernel (this follows from \citep[Prop. 1]{Carmeli2010}, using as feature operator the dot product with respect to $e_t$, where $e_t$ is defined in \cref{appendix:proof_K0_general})
\begin{talign*}
(K_0(x,y))_{tt}
& =\sum_{r=1}^d (\partial^r_x \partial^r_y K(x,y))_{tt}+ l_{tr}(y)\partial^r_x (K(x,y))_{tt} \\
& \quad + l_{tr}(x) \partial_y^r (K(x,y))_{tt} + l_{tr}(x) l_{tr}(y) (K(x,y))_{tt} \qquad \forall t \in [T].
\end{talign*}
In particular since $K$ is bounded with bounded derivatives, and
\begin{talign*}
    \Pi_t\left[| l_{tr} |\right]
  + \Pi_t \left[| l_{tr} |^2 \right]   \leq 
    \sqrt{\Pi_t\left[| l_{tr} |^2\right]}
  + \Pi_t \left[| l_{tr} |^2 \right] \qquad \forall t \in [T], r \in [d]
\end{talign*}
then $\int_{\X} (K_0(x,x))_{tt} \mathrm{d} \Pi_t (x) 
 < \infty $
if  $\| \nabla_x \log \pi_t(x) \|_2$ is square integrable with respect to $\Pi_t$, and the result follows. 
\end{proof}

\subsection{Proof of \Cref{thm:closed_form_minimiser}}
\label{appendix:proof optimal function}

\begin{proof}

We want to find 
  \begin{talign*}
  & \argmin_{g \in \mathcal{H}_{K_0}} \Lvv_m(g,\beta) \\
    \text{ where } \qquad \Lvv_m(g,\beta)
    & 
    := \sum_{t=1}^T  \frac{1}{m_t} \sum_{j=1}^{m_t} (f_t(x_{tj}) - g_t(x_{tj}) - \beta_t)^2 +  \lambda \|g\|^2_{\mathcal{H}_{K_0}}.
\end{talign*}
Note that the objective is the same as that in \Cref{eq:Lvv_m}, with the only difference being that the first input is now a function as opposed to the parameter value parameterising this function. We will abuse notation by using the same mathematical expression for both objectives.

By \citet[Section 2.1]{ciliberto2015_convex_learning_multitasks_and_structure}, any  solution of the minimization problem has the form $\hat g(\cdot) \defn \sum_{t'=1}^T\sum_{j'=1}^{m_{t'}} K_0(\cdot \,, x_{t'j'}) \theta_{t'j'}$.
Subbing this solution into $\Lvv_m(g,\beta)$ yields
 \begin{talign*}
  \Lvv_m (\hat g,\beta)  &= \sum_{t=1}^T  \frac{1}{m_t} \sum_{j=1}^{m_t} (f_t(x_{tj}) - ( \sum_{t'=1}^T\sum_{j'=1}^{m_{t'}} K_0(x_{tj},x_{t'j'})_t \theta_{t'j'} - \beta_t)^2 
  \\ & \qquad + \lambda  \sum_{t',t''=1}^T \sum_{j'=1}^{m_{t'}}\sum_{j''=1}^{m_{t''}} \theta_{t'j'}^T K_0(x_{t'j'},x_{t''j''})\theta_{t''j''}\\
  &= \lambda  \sum_{t',t''=1}^T \sum_{j'=1}^{m_{t'}}\sum_{j''=1}^{m_{t''}} \theta_{t'j'}^T K_0(x_{t'j'},x_{t''j''})\theta_{t''j''}
  \\
  & \qquad +
  \sum_{t=1}^T  \frac{1}{m_t} \sum_{j=1}^{m_t} \Big(y_{tj}^2 + (\sum_{t'=1}^T\sum_{j'=1}^{m_{t'}} K_0(x_{tj},x_{t'j'})_t \theta_{t'j'})^2 \\
  & \qquad \qquad \qquad \qquad \qquad \qquad -2 \sum_{t'=1}^T\sum_{j'=1}^{m_{t'}} y_{tj}
  K_0(x_{tj},x_{t'j'})_t \theta_{t'j'} \Big),
 \end{talign*}
 where $y_{tj} \defn f_t(x_{tj})-\beta_t$. 
 The problem thus becomes a minimization problem over the coefficients $\theta$,
 \begin{talign*}
  \argmin_{\theta \in \R^{| \mathcal D |}} 
  \lambda & \sum_{t',t''=1}^T \sum_{j'=1}^{m_{t'}}\sum_{j''=1}^{m_{t''}} \theta_{t'j'}^T K_0(x_{t'j'},x_{t''j''})\theta_{t''j''} \\
  & -  2 \sum_{t,t'=1}^T \frac{1}{m_t} \sum_{j=1}^{m_t}
  \sum_{j'=1}^{m_{t'}}
  \theta_{t'j'}^T K_0(x_{t'j'},x_{tj})_{\cdot t} y_{tj} 
  + 
  \sum_{t=1}^T \frac{1}{m_t} \sum_{j=1}^{m_t} y_{tj}^2\\
  &+
  \sum_{t,t',t''=1}^T\sum_{j=1}^{m_{t}}\sum_{j'=1}^{m_{t'}}
  \sum_{j''=1}^{m_{t''}}
    \theta_{t'j'}^T K_0(x_{t'j'},x_{tj})_{\cdot t} \frac{1}{m_t} K_0(x_{tj},x_{t''j''})_{t \cdot}\theta_{t'' j''} .
 \end{talign*}
Since the quadratic terms are positive definite, the resulting objective is a convex function of $\theta$,
thus, by differentiating it, we obtain that the solution $\theta$ is the solution to
 \begin{talign*}
\sum_{t'=1}^T \sum_{j'=1}^{m_{t'}} \Big( \sum_{t=1}^{T}\frac{1}{m_t} &  \sum_{j=1}^{m_t}  
     K_0(x_{t''j''},x_{tj})_{\cdot t}  K_0(x_{tj},x_{t'j'})_{t \cdot} + \lambda  K_0(x_{t''j''},x_{t'j'}) \Big) \theta_{t'j'} \\
    &=\sum_{t}^T \frac{1}{m_t} \sum_{j=1}^{m_t}    K_0(x_{t''j''
},x_{tj})_{\cdot t} (f_t(x_{tj})-\beta_t), \qquad \forall t'' \in [T], j'' \in [m_T].
\end{talign*}

Finally, generalising the scalar case, we say that a matrix-valued reproducing kernel $K_0$ is strictly positive definite if for any finite set of $\gamma_s \in \R^T$ and distinct points $y_s \in \R^d$   we have
$ \sum_{s,\ell} \gamma^\top_s K_0(y_s,y_\ell) \gamma_\ell =0$ implies each $\gamma_s$ is zero -- this means that the mean embedding of $K_0$ is injective (or characteristic) over the set of linear functionals of the form $\delta_y^\gamma:f \mapsto \sum_t f_t(y) \gamma_t$.
It follows that the map $\R^{\mathcal D_x} \times \R^T \to \R^{\mathcal D_x} \times \R^T $, where $\mathcal D_x =\{x_{1j},\ldots , x_{T m_T} \}$, defined as 
$$\theta \mapsto \left( \sum_{t=1}^T \sum_{j=1}^{m_t} K_0(x_{1j},x_{tj}) \theta_{tj} \,, \ldots, \sum_{t=1}^T \sum_{j=1}^{m_t} K_0(x_{Tm_T},x_{tj}) \theta_{tj}  \right) $$
is injective between vector spaces of the same dimension, and thus invertible by the rank theorem. 
Hence, since by above the quadratic term is positive definite, the linear system may be inverted to find $\theta^*$.
\end{proof}


\section{Alternative Constructions}\label{appendix:alternative_constructions}

 In this third appendix, we will now provide alternative constructions to those presented in the main text. First, in \Cref{appendix:2ndordervvCVs}, we present kernel-based \gls{vvCV}s derived from the second order Langevin Stein operator. Then, in \Cref{appendix:alternative_constructions_general_discussion} and \Cref{appendix:poly_vvCVs} we point out how these constructions can lead to polynomial-based \gls{vvCV}s.


\subsection{Kernel-based \gls{vvCV}s from Second-Order Langevin Stein Operators} \label{appendix:2ndordervvCVs}

The \emph{Langevin Stein operator} can also be adapted to apply to the derivative of twice differentiable
scalar-valued functions $u:\X \rightarrow \R$, in which case it is called the \emph{second-order Langevin Stein operator}:
\begin{talign}\label{eq:Langevin_sv}
\L'[u](x):= \Delta_{x} u(x) + \nabla_{x} u(x) \cdot \nabla_{x} \log \pi(x),
\end{talign}
where $\Delta_x = \nabla_x \cdot \nabla_x$.

In this section we will consider the second-order Langevin Stein operator which acts on scalar-valued functions. The following theorem provides a characterisation of the class of \gls{vvfunctions} obtained when applying this operator to functions in a \gls{vvRKHS}.
\begin{theorem}\label{thm:K0_kernel_secondorder}
Consider $\mathcal{H}_K$ which is a \gls{vvRKHS} with \gls{mvkernel} $K:\X \times \X \rightarrow \R^{T \times T}$, and suppose that $K \in C^{2,2}(\X \times \X)$. Furthermore, for suitably regular \gls{vvfunctions} $u=(u_1,\ldots,u_T):\X \rightarrow \R^T$ define the differential operator
\begin{talign*}
\Svv[u] = (\L'_{\Pi_1}[u_1],\ldots, \L'_{\Pi_T}[u_T])^\top.
\end{talign*}
 Then, the image of $\mathcal{H}_K$ under $\Svv$ is a \gls{vvRKHS} with reproducing kernel $K_0:\X \times \X \rightarrow \R^{T \times T}$:
\begin{talign*}
(K_0(x,y))_{tt'} 
& =
\sum_{r,s=1}^d  \partial^{ss}_x \partial^{rr}_y (K(x,y))_{tt'} + l_{t'r}(y)\partial^{ss}_x \partial^{r}_y (K(x,y))_{tt'} \\
&
\qquad + l_{ts}(x) \partial^s_x \partial^{rr}_y (K(x,y))_{tt'}+  l_{ts}(x)  l_{t'r}(y) \partial^{s}_x \partial^r_y (K(x,y))_{tt'} \qquad \forall t,t' \in [T].
\end{talign*}
\end{theorem}
We note that this theorem is very similar to \Cref{thm:K0_kernel_firstorder}, and recovers the kernel of \cite{Barp2018} when $T=1$, provided we use the manifold  analog of \cref{eq:Langevin_sv}.
Indeed, one advantage of \cref{eq:Langevin_sv} is that the associated \Cref{thm:K0_kernel_secondorder}  can be easily extended to manifolds (more generally, we can obtain a similar result for any generators of measure-preserving diffusion given in Corollary 5.3 of \citet{barp2021unifying}). 
However, one particular disadvantage of this construction from a computational viewpoint is that it requires higher-order derivatives of the kernel $K$. It also requires the evaluation of a double sum, which significantly increases computational cost relative to our construction in the main text. For this reason, we did not explore this construction in more details.

\begin{proof} 
We proceed as for the proof of \Cref{thm:K0_kernel_firstorder} and shall derive a feature map for $K_0$.
Recall  that $g = \Svv[u] = (\L_{\Pi_1}'[u_1],\ldots,\L_{\Pi_T}'[u_T])^\top$, where $\L_{\Pi_i}'$ is the second-order Stein operator, which maps scalar functions to scalar functions. Here $u$ belongs to a \gls{RKHS} of $\R^T$-valued functions with matrix kernel $K$.
From the differentiability assumption on $K$, we have $\H_K \subset C^2$, i.e., it is a space of twice continuously differentiable functions.
Note that
 (here $\partial^{jj} = \partial^j \partial^j = \frac{\partial^2}{\partial x_j \partial x_j}$)
\begin{talign*}
\metric{\partial^{jj}_x K(\cdot, x) e_t }{u}_{\H_K} = \partial^{jj} u_t (x) \defn
\frac{\partial^2 u_t}{\partial x_j \partial x_j}(x) \qquad \forall t \in [T],
\end{talign*}
where $e_t$ is the $t^{th}$ standard basis  vector of $\R^T$ as before.
Thus
\begin{talign*}
\L_{\Pi_t}'[u_t](x) & = \Delta_x u_t(x) + \nabla_x \log \pi_t(x) \cdot \nabla_x u_t(x) \\
&=  \sum_{s=1}^d \partial^{ss} u_t(x) + \sum_{s=1}^d l_{ts} (x)  \partial^s  u_t (x) 
 \\
 &= 
\sum_{s=1}^d \metric{\partial^{ss}_x K(\cdot, x) e_t }{u}_{\H_K} +
\sum_{s=1}^d \metric{ l_{ts} (x) \partial^{s}_x K(\cdot, x) e_t }{u}_{\H_K} \\
& =
\sum_{s=1}^d \metric{\partial^{ss}_x K(\cdot, x) e_t + l_{ts} (x) \partial^{s}_x K(\cdot, x) e_t}{u}_{\H_K} \qquad \qquad \qquad \qquad \forall t \in [T].    
\end{talign*}
Hence 
\begin{talign*}
\Svv[u](x)= \begin{pmatrix}  
\metric{\sum_{s=1}^d \partial^{ss}_x K(\cdot, x) e_1 +l_{1s} (x) \partial^{s}_x K(\cdot, x) e_1}{u}_{\H_K}\\
\vdots\\
\metric{\sum_{s=1}^d \partial^{ss}_x K(\cdot, x) e_T +l_{Ts} (x) \partial^{s}_x K(\cdot, x) e_T}{u}_{\H_K}
\end{pmatrix} \in \R^T.
\end{talign*}
Note that for each $x \in \X$, each component of the above  is  a bounded linear operator $\H_K \to \R$ (i.e., the map $u \mapsto (\Svv(u)(x))_s \in \R$ to the $s$-component is a bounded linear operator),
then 
we have obtained a  a feature map, i.e., a map 
$\gamma: \X \to \mathcal B(\H_K,\R^T)$, where $\mathcal B(\H_K,\R^T)$ denotes the space of bounded linear maps from $\H_K$ to $\R^T$.
Specifically
\begin{talign*}
\gamma(x) \defn  \Svv[\cdot](x) \in \mathcal B(\H_K,\R^T)  .
\end{talign*}
In particular, as before 
\begin{talign*}
K_0(x,y) \defn \gamma(x) \circ \gamma(y)^* \in \mathcal B(\R^T,\R^T)
\end{talign*}
will thus be the kernel associated to the ``feature operator'' $\Svv:\H_K \to \H_{K_0}$.
Recall that  $\gamma(y)^* \in \mathcal B(\R^T,\H_K)$ is the adjoint map to $\gamma(y)$, i.e., it satisfies for any $a \in \R^T, u \in \H_K$:
\begin{talign*}
 \metric{ \gamma(y)^*[a]}{u}_{\H_K} = \gamma(y)[u] \cdot a.
\end{talign*}
From this we obtain
\begin{talign*}
 \gamma(y)^* : a \mapsto  \sum_{r=1}^d \sum_{t=1}^T a_t \left(
 \partial^{rr}_y K(\cdot, y) e_t +l_{tr} (y) \partial^{r}_y K(\cdot, y) e_t \right) \in \H_K.
\end{talign*}
From $K_0(x,y)a = \gamma(x) \circ  \gamma(y)^*[a]$ for all $a \in \R^T$ and the above expressions we can finally calculate $K_0$. We  have
\begin{talign*}
K_0(x,y)a = \Svv[\gamma(y)^*a](x)= \begin{pmatrix}  
\metric{\sum_{s=1}^d \partial^{ss}_x K(\cdot, x) e_1 +l_{1s} (x) \partial^{s}_x K(\cdot, x) e_1}{\gamma(y)^*a}_{\H_K}\\
\vdots\\
\metric{\sum_{s=1}^d \partial^{ss}_x K(\cdot, x) e_T +l_{Ts} (x) \partial^{s}_x K(\cdot, x) e_T}{\gamma(y)^*a}_{\H_K}
\end{pmatrix}.
\end{talign*}
 We obtain that $K_0(x,y)a $ is a vector with components:
  \begin{talign*}
  (K_0(x,y)a)_t  = \sum_{r,s=1}^d \sum_{t'=1}^T a_{t'}\Big( & (\partial^{ss}_x \partial^{rr}_y K(x,y))_{tt'} + l_{t'r}(y)(\partial^{ss}_x \partial^{r}_y K(x,y))_{tt'} \\
   & + l_{ts}(x) (\partial^s_x \partial^{rr}_y K(x,y))_{tt'}+ l_{ts}(x)  l_{t'r}(y)
(\partial^{s}_x \partial^r_y K(x,y))_{tt'}\Big) \quad \forall t \in [T].
  \end{talign*}
Thus the components of $K_0(x,y) \in \R^{T \times T}$ are
\begin{talign*}
(K_0(x,y))_{tt'} 
& =
\sum_{r,s=1}^d  (\partial^{ss}_x \partial^{rr}_y K(x,y))_{tt'} + l_{t'r}(y)(\partial^{ss}_x \partial^{r}_y K(x,y))_{tt'} \\
&
+ l_{ts}(x) (\partial^s_x \partial^{rr}_y K(x,y))_{tt'}+  l_{ts}(x)  l_{t'r}(y) (\partial^{s}_x \partial^r_y K(x,y))_{tt'} \qquad \forall t,t'\in [T].
\end{talign*}
\end{proof}

Analogously to the \gls{mvkernel} in \Cref{thm:K0_kernel_firstorder}, there are several cases of practical interest. The first is when $K(x,y) = B k(x,y)$ is a separable kernel, in which case:
\begin{talign*}
     (K_0(x,y))_{tt'} &= B_{tt'} \sum_{r,s=1}^d \partial^{ss}_x \partial^{rr}_y k(x, y)  + l_{t'r}(y) \partial^{ss}_x \partial^r_y k(x,y)   \nonumber\\
     &\quad \qquad + l_{ts}(x) \partial^s_x \partial^{rr}_y k(x, y) + l_{ts}(x) l_{t'r}(y) \partial^s_x \partial^r_y k(x, y) \qquad \forall t,t' \in [T].
\end{talign*}
The second is when $K$ is separable and $\Pi_1 = \ldots = \Pi_T$, in which case $l_r(x) := l_{1r}(x) = \ldots = l_{Tr}(x)$ $\forall r \in [d]$ and:
\begin{talign*}
    (K_0(x,y))_{tt'} & = B_{tt'} \sum_{r,s=1}^d   \partial^{ss}_y \partial^{rr}_x k(x,y)  +  l_r(x) ~ \partial^{ss}_y \partial^r_x k(x,y)   \nonumber \\
     &\quad \qquad + l_s(y)  ~ \partial^s_y \partial^{rr}_x k(x,y) + l_s(y) l_r(x) \partial^s_y \partial^r_x k(x,y) \qquad \forall t,t' \in [T]. 
\end{talign*}

\subsection{Alternative Constructions beyond Kernels}\label{appendix:alternative_constructions_general_discussion}

Although kernels are a natural way of constructing functions for multi-task problems, it is also possible to generalise constructions based on other parametric families such as polynomials or neural networks. We will not explore this avenue in detail in the present paper, but now provide brief comments on how such generalisations could be obtained.

Firstly, $u_\theta$ could be based on any additive model such as a polynomial or wavelet expansion. In that case, it is straightforward to construct \gls{vvCV}s with a separable structure as follows: 
\begin{talign}
\label{eq:vv-linear-CV}
(u_\theta(x))_t   = \sum_i \sum_{t' =1}^T B_{tt'} \theta_i \phi_i(x), \quad (g_\theta(x))_t  = \sum_{i} \sum_{t' =1}^T B_{tt'} \theta_{i} \Ssv_{\Pi_t}[\phi_i(x)]  \quad \forall t \in [T], 
\end{talign}
where $B \in S^T_+$ and $\phi_i:\X \rightarrow \R$ is a (sufficiently regular) basis function. In particular, taking the basis functions to be of the form $x^{\alpha}$ for $\alpha \in \N^d$ recovers the polynomial-based \gls{CV}s of \cite{mira2013_zerovariance_MCMC}. We also note that any model of this form leads to a quadratic \gls{MC} variance objective, whose solution can be obtained in closed form under mild regularity conditions on the basis functions.

Secondly, we could use non-linear models for $u_\theta$. In that case, one approach would be to use a separable structure of the form:
\begin{talign}
(u_\theta(x))_t   = \sum_{t' =1}^T B_{tt'} \phi_\theta(x), \quad (g_\theta(x))_t  = \sum_{t'=1}^T B_{tt'} \Ssv_{\Pi_t}[\phi_{\theta}(x)]  \quad \forall t \in [T]. \label{eq:vv-nonlinear-CV}
\end{talign}
where $\phi_\theta(x)$ is a non-linear function of the parameters $\theta$. The above is a generalisation of the neural networks-based \gls{CV}s  of \cite{zhu2018neural_CV,Si2020} whenever $\phi_\theta$ is a neural network. Unfortunately the \gls{MC} variance objective will usually be non-convex in those cases, and we therefore have no guarantees of recovering the optimal parameter value when using most numerical optimisers.


\subsection{Polynomial \gls{vvCV}s}\label{appendix:poly_vvCVs}

In \Cref{appendix:alternative_constructions_general_discussion}, we have discussed a construction for \gls{vvCV}s based on polynomials which recovers the work of \cite{mira2013_zerovariance_MCMC}. However, it is also possible to obtain polynomial-based \gls{vvCV}s directly through our kernel constructions in \Cref{thm:K0_kernel_firstorder} and \Cref{appendix:2ndordervvCVs}. In particular, one option would be to take $K(x,y) = B k(x,y)$ where $B \in S_+^T$ and $k(x,y) = (x^\top y +c)^l$ where $c \in \R$ and $l\in \N$. Firstly, using the first-order Langevin Stein operator and setting $l=1$, we obtain:
\begin{talign*}
   \left( K_0(x,y) \right)_{tt'} &= B_{tt'} \sum_{r=1}^d \left[ 1 + l_{t'r}(y) y_r + l_{tr}(x) x_r + l_{tr}(x) l_{t'r}(y) \left( x^\top y +c \right) \right] \qquad \forall t \in [T].
\end{talign*}
Similarly when $l=2$, we get:
\begin{talign*}
   \left( K_0(x,y) \right)_{tt'} &=  B_{tt'} \sum_{r=1}^d \Big[ 2 x_r y_r + 2 \left( x^\top y +c \right) + 2 y_r  l_{t'r}(y) \left(  x^\top y +c \right)  \\
   & \qquad \qquad + 2 x_r l_{tr}(x) \left( x^\top y +c \right)  + l_{tr}(x) l_{t'r} \left( x^\top y  +c \right)^2  \Big] \qquad \forall t \in [T].
\end{talign*}
These two choices were considered in the experiments in \Cref{sec:experiments}. An alternative would be to consider this same kernel, but using the construction based on second-order Langevin Stein operators. Again, taking $l=1$, we obtain: 
\begin{talign*}
(K_0(x,y))_{tt'} = \sum_{r=1}^d l_{tr}(x) l_{t'r}(y) B_{tt'}  \qquad \forall t \in [T].
\end{talign*}
Similarly, when $l=2$, we get:
\begin{talign*}
    (K_0(x,y))_{tt'} &= B_{tt'}\Big[ 4 \left(d  +    \sum_{r=1}^d l_{t'r}(y)y_r + l_{tr}(x) x_r\right)  \\ 
    &\quad + 2 \left( \sum_{r=1}^d  l_{tr}(x) l_{t'r}(y) \left( x^\top y + c \right) + \sum_{r,s=1}^d  l_{ts}(x) l_{t'r}(y)  x_r y_s  \right) \Big]  \qquad \forall t \in [T].
\end{talign*}


\section{Implementation Details}\label{appendix:implementation}

In this appendix, we focus on implementation details which may be helpful for implementing the algorithms in the main text. Firstly, in \Cref{appendix:kernel_derivatives} we derive the derivatives of several common kernels; this is essential for the implementation of Stein reproducing kernels. Then, in \Cref{appendix:hyperparameters}, we provide details on how to select hyperparameters. Finally, in \Cref{appendix:convex_optim_B}, we discuss how to turn the problem of estimating $B$ from data into a sequence of convex optimisation problems.


\subsection{Kernels and Their Derivatives}\label{appendix:kernel_derivatives}

We now provide details of all the kernels used in the paper, as well as expressions for their derivatives.

\paragraph{Polynomial Kernel}
The polynomial kernel $k_l(x,y) = (x^\top y + c )^l$ with constant $c\in \R$ and power $l \in \N$ has derivatives given by
\begin{talign*}
   \nabla_x k_l(x,y) & = l (x^\top y +c)^{l-1} y,  \quad   \nabla_y k_l(x,y) = l (x^\top y +c)^{l-1} x,  \\
    \nabla_x \cdot \nabla_y k_l(x,y) &= \sum_{j=1}^d \frac{\partial^2}{\partial x_j\partial y_j} k_l(x,y) = \sum_{j=1}^d \frac{\partial}{\partial x_j} \left[ l(x^\top y + c)^{l-1} x_j \right] \\
    &= \sum_{j=1}^d  l (l-1) (x^\top y +c)^{l-2} y_j x_j + l (x^\top y + c)^{l-1}  \\
    &= l(l-1) (x^\top y + c)^{l-2} x^\top y + d l (x^\top y + c)^{l-1}.
\end{talign*}

\paragraph{Squared-Exponential Kernel}
The squared-exponential kernel (sometimes called Gaussian kernel) $k(x,y) = \exp(-\frac{\|x-y\|_2^2}{2\lambda})$ with lengthscale $\lambda > 0$ has derivatives given by
\begin{talign*}
    \nabla_x k(x,y) &= - \frac{(x-y)}{\lambda} k(x,y), \qquad
    \nabla_y k(x,y) = \frac{(x-y)}{\lambda} k(x,y),\\
    \nabla_x \cdot \nabla_y k(x,y) &= \sum_{j=1}^d \frac{\partial^2}{\partial y_j \partial x_j}  k(x,y)  = \sum_{j=1}^d \frac{\partial}{\partial y_j} \left[ -\frac{(x_j-y_j)}{\lambda}k(x,y) \right] \\
    & = \sum_{j=1}^d \left[ \frac{1}{\lambda} - \frac{(x_j - y_j)^2}{\lambda^2} \right] k(x,y)
    = \left[ \frac{d}{\lambda} - \frac{(x - y)^\top (x - y)}{\lambda^2} \right] k(x,y).
\end{talign*}

\paragraph{Preconditioned Squared-Exponential Kernel}
Following \citet{oates2017_CF_for_MonteCarloIntegration}, we also considered a preconditioned squared-exponential kernel:
\begin{talign*}
    k(x,y) &= 
    \frac{1}{(1+\alpha \|x\|_2^2)(1+\alpha\|y\|_2^2)} \exp\left(-\frac{\|x-y\|_2^2}{2\lambda^2}\right).
\end{talign*}
with lengthscale $\lambda>0$ and preconditioner parameter $\alpha>0$. This kernel has derivatives given by:
\begin{talign*}
\nabla_x k(x,y) &= 
    \left[ \frac{- 2 \alpha x}{1+\alpha \|x\|_2^2} - \frac{(x-y)}{\lambda^2} \right] k(x,y), \quad
    \nabla_y k(x,y) = 
    \left[ \frac{- 2 \alpha y}{1+\alpha \|y\|_2^2} + \frac{(x-y)}{\lambda^2}  \right] k(x,y), \\
    \nabla_x \cdot \nabla_y k(x,y) 
    & = 
    \sum_{j=1}^d \frac{\partial^2 }{\partial x_j \partial y_j} k(x,y) 
    = 
    \sum_{j=1}^d \frac{\partial}{\partial y_j}\left[  \left( \frac{- 2 \alpha x_j}{1+\alpha \|x\|_2^2} - \frac{(x_j-y_j)}{\lambda^2} \right) k(x,y)  \right]  \\
    &=\sum_{j=1}^d \left( \frac{1}{\lambda^2} k(x,y) + \left[  \frac{- 2 \alpha x_j}{1+\alpha \|x\|_2^2} - \frac{(x_j-y_j)}{\lambda^2}  \right] \frac{\partial}{\partial y_j} k(x,y) \right)  \\
    &= \sum_{j=1}^d \left( \frac{1}{\lambda^2} k(x,y) + \left[ \frac{- 2 \alpha x_j}{1+\alpha \|x\|_2^2} - \frac{(x_j-y_j)}{\lambda^2}   \right] \left[ \frac{- 2 \alpha y_j}{1+\alpha \|y\|_2^2} + \frac{(x_j-y_j)}{\lambda^2}    \right] k(x,y)    \right)  \\
    &= k(x,y) \left[ \frac{4\alpha^2 x^\top y}{(1+\alpha\|x\|_2^2)(1+\alpha\|y\|_2^2)} + \frac{2\alpha (x-y)^\top y}{\lambda^2(1+\alpha \| y \|_2^2)} 
    -  \frac{2\alpha(x-y)^\top x}{\lambda^2(1+\alpha\|x\|_2^2)} + \frac{d}{\lambda^2} - \frac{(x-y)^\top(x-y)}{\lambda^4}  \right].   
\end{talign*}

\paragraph{Product of Kernels} Finally, some of our examples will also use products of well-known kernels. Consider the kernel $k(x,y) = \prod_{j=1}^d k_j(x_j, y_j)$. The derivatives of this kernel can be expressed in terms of the components of the product and their derivates as follows:
\begin{talign*}
\nabla_x k(x,y) 
&= 
\left(\frac{\partial k_1(x_1,y_1)}{\partial x_1}  \prod_{j \neq 1} k_j(x_j, y_j), \ldots, \frac{\partial k_d(x_d,y_d)}{\partial x_d} \prod_{j \neq d} k_j(x_j, y_j)\right)^\top \\
\nabla_y k(x,y) 
&=
\left(\frac{\partial k_1(x_1,y_1)}{\partial y_1} \prod_{j \neq 1} k_j(x_j, y_j),\ldots, \frac{\partial k_d(x_d,y_d)}{\partial y_d}  \prod_{j \neq d} k_j(x_j, y_j) \right)^\top \\
 \nabla_y \cdot \nabla_x k(x, y) 
 &=
 \sum_{j=1}^d \frac{\partial^2}{\partial x_j \partial y_j} k(x,y) = \sum_{j=1}^d \frac{\partial}{\partial y_j} \left( \frac{\partial k_j(x_j, y_j)}{\partial x_j}  \prod_{i\neq j} k_i(x_i, y_i) \right)  \\
  &=  \sum_{j=1}^d \left[ \frac{\partial^2 k_j(x_j, y_j)}{\partial y_j \partial x_j}  \prod_{i\neq j} k_i(x_i, y_i) \right].
\end{talign*}


\subsection{Hyper-parameters Selection} \label{appendix:hyperparameters}

Most kernels (whether scalar- or matrix-valued) will have hyperparameters which we will have to select. For example, the squared-exponential kernel will often have a lengthscale or amplitude parameter, and these will have a significant impact on the performance. 

We propose to select kernel hyperparameters through a marginal likelihood objective by noticing the equivalence between the optimal \gls{vvCV} based on the objective in \Cref{eq:Lvv_m} and the posterior mean of a zero-mean Gaussian process model with covariance matrix $K_0(x,y)$; see \cite{oates2017_CF_for_MonteCarloIntegration} for a discussion in the sv-CV case. Unfortunately, computing the marginal likelihood in the general case can be prohibitively expensive due to the need to take inverses of large kernel matrices; the exact issue we were attempting to avoid through the use of the stochastic optimisation approaches. For simplicity, we instead maximise the marginal likelihood corresponding to $B=I_T$:
\begin{talign*}
\nu^* := \argmax_{\nu} - \frac{1}{2} \sum_{t=1}^{T} & \left( \sum_{j,j'=1}^{m_t} f_{t}(x_{tj}) (K_{\Pi_t}(\nu)+\lambda I_{m_t})^{-1}_{jj'} f_{t}(x_{tj'}) +\log \det[K_{\Pi_t}(\nu) + \lambda I_{m_t}] \right).
\end{talign*}
where $K_{\Pi_t}(\nu)$ is a matrix with entries $K_{\Pi_t}(\nu)_{ij} = k_{\Pi_t}(x_{ti}, x_{tj}; \nu)$ where $k_{\Pi_t}$ is a Stein reproducing kernel of the form in \Cref{eq:k_0} specialised to $\Pi_t$ which has hyperparameters given by some vector $\nu$. This form is not optimal when $B \neq I_T$, but we found that it tend to perform well in our numerical experiments. The regularisation parameter $\lambda$ can also be selected through the marginal likelihood. However, in practice we are in an interpolation setting and therefore choose $\lambda$ as small as possible whilst still being large enough to guarantee numerically stable computation of the matrix inverses above.


\subsection{Convex Optimisation for Estimating $B$}\label{appendix:convex_optim_B}

As discussed in \Cref{sec:learn_B_jointly_case1}, estimating the matrix $B$ for a separable kernel from data leads to a non-convex optimisation problem. Thankfully, we can approximate the optimum using a sequence of convex problems by extending the work of \citet{dinuzzo2011_learning_outputkernel_with_blockcoordinateDescent,ciliberto2015_convex_learning_multitasks_and_structure} together with \Cref{thm:closed_form_minimiser} above. 
For this, we will require that the kernel $K_0$ is separable, and shall thus restrict ourselves to the case where we have a single target distribution (i.e. special case II).

\begin{theorem}\label{thm:convex problems minimisation}
Suppose that  $\Pi_t =\Pi$ for $t \in [T]$ and $K(x,y)=B k(x,y)$ so that $K_0(x,y)=B k_0(x,y)$ where $k_0$ is defined in \Cref{eq:k_0}.
Then the following objective is convex in $(\theta,\beta,B)$ for any value of $\delta>0$:
\begin{talign*}
   \bar{L}^{\text{vv}}_{m,\delta}(\theta,\beta,B) =  \Jvv_{m}(\theta,\beta,I_T) + \lambda  \sum_{t,t'=1}^T \sum_{j=1}^{m_t} \sum_{j'=1}^{m_{t'}}\mathrm{Tr}\left[B^{\dagger} \left( k_0(x_{tj},x_{t'j'}) \theta_{tj}  \theta_{t'j'}^{\top}  + \delta^2 I_T \right)\right]  + \|B\|^2,
\end{talign*}
and for each $\beta$ and any 
sequence $\delta_\ell \rightarrow 0$, the associated 
sequence of minimisers $(\theta_\ell, B_\ell)$ converges to $(\theta_*,B_*)$ s.t., 
$(\theta_* B_*^{\dagger},B_*)$
minimises the objective in \Cref{eq:vv_CV_learnB_1D1P_raw_obj}.
\end{theorem}
\begin{proof}
Since the kernel $K_0$ is separable, the objective \cref{eq:vv_CV_learnB_1D1P_raw_obj}
may be written in the form of 
  \citet[Problem $(\mathcal Q)$]{ciliberto2015_convex_learning_multitasks_and_structure}.
  Has shown therein, $\sum_{t,t'=1}^T \sum_{j=1}^{m_t} \sum_{j'=1}^{m_{t'}}\text{Tr} \left[B^{\dagger} \left( k_0(x_{tj},x_{t'j'}) \theta_{tj}  \theta_{t'j'}^{\top}  \right)\right]$ is jointly convex in $B$ and $\theta$, 
  and since the first term in $\Lvv_{m,\delta}(\theta,\beta,B)$ is convex in $\beta$ and $\theta$ jointly, $\Lvv_{m,\delta}(\theta,\beta,B)$ is jointly convex in $(\theta, B, \beta)$. Moreover, 
  by Theorem~3.1 \& 3.3  in 
\citep{ciliberto2015_convex_learning_multitasks_and_structure},  when $\delta \rightarrow 0$, $(\theta, B)$ converges in Frobenius norm  to $(\theta_*, B_*)$, where
$(\theta_* B^{\dagger}_*, B_*)$ a minimiser of
\cref{eq:vv_CV_learnB_1D1P_raw_obj}, where $B^{\dagger}_*$ denotes the pseudoinverse of $B_*$.
\end{proof}
This theorem could therefore be used to construct an approach based on convex optimisation algorithms which are used iteratively for a decreasing sequence of penalisation parameters in order to converge to an optimum approaching the global optimum. However, this approach is limited to the case where all distributions are identical, and is hence not as widely applicable as Algorithm \ref{alg:vvCV_block}.

\section{Additional Details for the Experimental Study}\label{appendix:additional_experiments}

This last Appendix provides additional experiments including:
an illustration plot of matrix-valued Stein reproducing kernel in \Cref{appendix:illustration_mvKernel}; a synthetic example from \cite{South2022} when the Stein kernel matches the smoothness of integrands in \Cref{appendix:south_experiments}; extra experiments for physical modelling of waterflow when having unbalanced datasets in \Cref{appendix:unbalanced_vvCV_borehole}.

Meanwhile, additional details of our numerical experiments in \Cref{sec:experiments} of the main paper are provided: multifidelity univariate step functions in \Cref{appendix:multifidelity_uni_step}; multifidelity modelling of waterflow in \Cref{appendix:multifidelity_modelling}, model evidence for dynamic systems in \Cref{appendix:TI} and Bayesian inference of Lotka-Volterra system in \Cref{appendix:Lotka}.

\subsection{Illustration of Matrix-valued Stein Kernels}
\label{appendix:illustration_mvKernel}

An illustration of matrix-valued Stein kernels $K_0$ is demonstrated in \Cref{fig:K0_plots} for the case $T=2$. As observed, the choice of kernel $k$ can have significant impacts on $K_0$. Moreover, $K_0$ possesses a well-known property of Stein kernels: even when $k$ is translation-invariant (see the top row) this may not be the case for $K_0$. This is due to the fact that $K_0$ depends on $l$. Finally, we can also observe that the two outputs of $1^\top K_0 (x,y)$ are correlated, a property which will be key when it comes to \gls{vvCV}s.

\begin{figure}[ht]
\vskip 0.1in
\begin{center}
\centerline{\includegraphics[width=0.8\columnwidth]{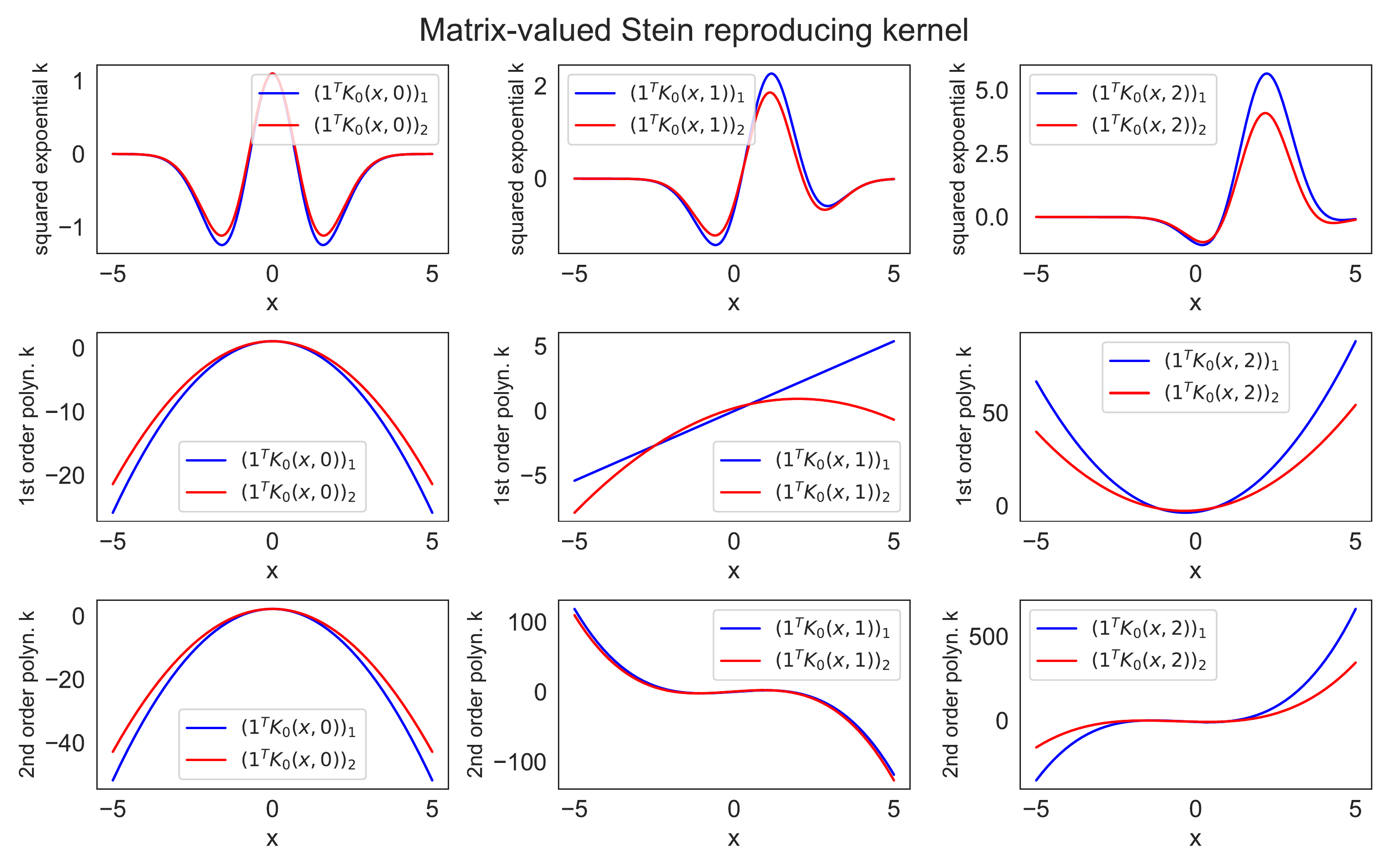}}
\caption{Illustration of a separable \gls{mvkernel} $K_0$ for $T=2$ through projections with $\bm{1} = (1,1)$. Here, $\Pi_1 = \mathcal{N}(0,1)$, $\Pi_2 = \mathcal{N}(0,1.25)$, $B_{11}=B_{22} =1$ and $B_{12}=B_{21} = 0.1$. The first row corresponds to taking $k$ to be a squared-exponential kernel, whereas the second and third row correspond to taking a polynomial kernel $k(x,y)=(x^\top y+1)^l$ with $l=1$ and $l=2$ respectively.}
\label{fig:K0_plots}
\end{center}
\vskip -0.2in
\end{figure}

\subsection{Additional Experiment: A Synthetic Example}
\label{appendix:south_experiments}
Here is a synthetic example selected from \cite{South2022} (denoted $f_2$), and to make the problem fit into our framework we introduced another similar integrand (denoted $f_1$): 
\begin{talign*}
&f_1(x) = 1.5 + x +1.5 x^2 + 1.75 \sin(\pi x)\exp(-x^2), \\
& f_2(x) = 1 + x + x^2 + \sin(\pi x) \exp(-x^2).
\end{talign*}
For this problem, we trained all \gls{CV}s through stochastic optimisation and use $m=(50,50)$ \gls{MC} samples. This synthetic example was originally used by \cite{South2022} to show one of the drawbacks of kernel-based \gls{CV}s, namely that the fitted CV will usually tend to $\beta$ in parts of the domain where we do not have any function evaluations. This phenomenon can be observed on the red lines in \Cref{fig:south2022_toyproblem} (left and center) which gives a \gls{CV} based on a squared-exponential kernel. This behaviour is clearly one of the biggest drawbacks of existing kernel-based approaches. However, the blue curve, representing a kernel-based \gls{vvCV} with separable kernel where $B$ was inferred through optimisation, partially overcomes this issue by using evaluations of both integrands, hence clearly demonstrating potential advantages of sharing function values across integration tasks. 

The right-most plot in \Cref{fig:south2022_toyproblem} presents several box plots for the sum of squared errors for each integration problem calculated over $100$ repetitions of the experiment. The different box plots show the impact of the difference in $\Pi_1$ and $\Pi_2$. As we observed, \gls{vvCV}s tend to outperform \gls{CV}s, although this difference in performance is more stark when $\Pi_2$ has a larger tail than $\Pi_1$. This reinforces the previous point, since a more disperse $\Pi_2$ means that the second integrand will be evaluated more often at more extreme areas of the domain, which will help obtain a better \gls{vvCV} by improving the fit at the tails of the distribution.

In this experiment, the choice of $k$ as a squared-exponential kernel was motivated by the fact that this makes $k_0$ infinitely differentiable, and hence matching the smoothness of both $f_1$ and $f_2$.

\begin{figure}[ht]
\vskip -0.1in
\begin{center}
\centerline{\includegraphics[width=\columnwidth]{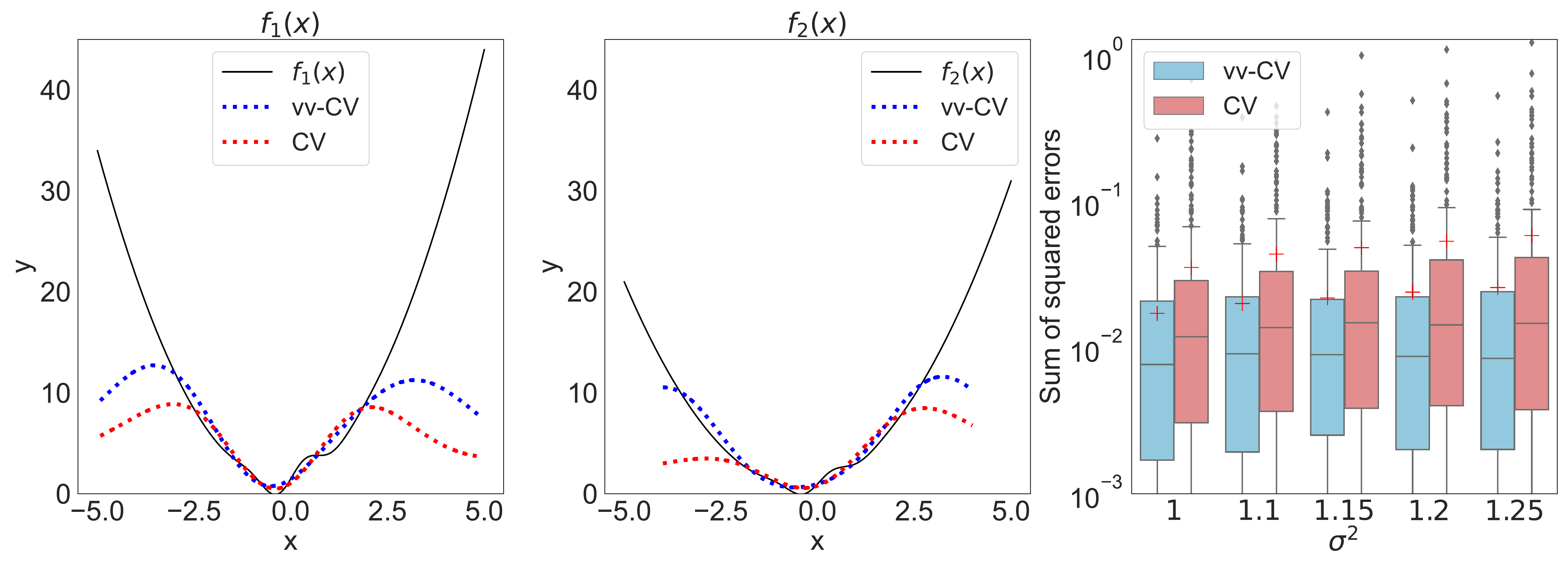}}
\caption{\emph{Numerical integration of problem from \cite{South2022}.} \emph{Left and center:} Illustration of $f_1$ and $f_2$, as well as the corresponding kernel-based \gls{CV}s and \gls{vvCV}s obtained through stochastic optimisation when $\Pi_1 = \mathcal{N}(0,1)$ and $\Pi_2 = \mathcal{N}(0,1.25)$. \emph{Right:} Sum of the squared errors in estimating $\Pi_1[f_1]$ and $\Pi_2[f_2]$. Here, $\Pi_1 = \mathcal{N}(0,1)$ whilst $\Pi_2 = \mathcal{N}(0,\sigma^2)$ where $\sigma^2 \in \{1,1.1,1.15,1.2,1.25\}$.}
    \label{fig:south2022_toyproblem}
\end{center}
\vskip -0.2in
\end{figure}

The experiment was replicated $100$ times for all methods. The exact details of the implementation are as follows.
\begin{itemize}
    \item \gls{CV} 
          \begin{itemize}
              \item Sample size: $50$.
              \item Hyper-parameter tuning: batch size $5$; learning rate $0.05$; total number of epochs $30$.
              \item Base kernel: squared exponential kernel
              \item Optimisation: $\lambda = 0.001$; batch size is $5$; learning rate is $0.001$; total number of epochs $400$.
          \end{itemize}
    \item \gls{vvCV} (estimated $B$) 
          \begin{itemize}
              \item Sample size: $(50, 50)$ from $(\Pi_1, \Pi_2)$ for $(f_1, f_2)$.
              \item Hyper-parameter tuning: batch size $5$ ($10$ in total for $(f_1, f_2)$); learning rate $0.05$; total number of epochs $30$.
              \item Base kernel: squared exponential kernel
              \item Optimisation: $B^{(0)}$ is initialized at the identity matrix $I_2$. $\lambda = 0.001$; batch size is $5$ ($10$ in total for $(f_1, f_2)$); learning rate is $0.001$; total number of epochs $400$.
          \end{itemize}
\end{itemize}

\subsection{Experimental details of Multi-fidelity Univariate Step Functions}
\label{appendix:multifidelity_uni_step}

The experiment is replicated $100$ times for all methods. Details of their implementation is given below:
\begin{itemize}
    \item Squared-exponential kernel

          \begin{itemize}
                \item CV 
                         \begin{itemize}
                            \item Sample size: $40$.
                            \item Base kernel: squared exponential kernel.
                            \item Hyper-parameter tuning:  batch size is  $10$; learning rate $0.02$; total number of epochs $15$.
                            \item Optimisation: $\lambda = 1e-5$; batch size is $10$; learning rate is $3e-4$; total number of epochs $400$.
                      \end{itemize}
                 \item vvCV (estimated B/fixed B) 
                      \begin{itemize}
                           \item Sample size: $(40,40)$ from $(\text{Normal}(0,1),    \text{Normal}(0,1))$ for $(f_L, f_H)$.
                           \item Hyper-parameter tuning:  batch size is  $5$ ($10$ in total for  $(f_L, f_H)$); learning rate $0.02$; total number of epochs $15$.
                          \item Base kernel: squared exponential kernel.
                          \item Optimisation: When $B$ is fixed, we set $B_{11} = B_{22} = 0.5, B_{12} = B_{21} = 0.01$; otherwise, $B^{(0)}$ is initialized at the identity matrix $I_2$. $\lambda = 1e-5$; batch size is  $5$ ($10$ in total for $(f_L, f_H)$); learning rate is $3e-4$; total number of epochs $400$.
                        \end{itemize}
          \end{itemize}
\item First-order polynomial kernel
            \begin{itemize}
            \item  \gls{CV}
                     \begin{itemize}
                         \item Sample size: $40$.
                         \item Base kernel: first order polynomial kernel.
                         \item Optimisation: $\lambda = 10^{-5}$; batch size is $10$; learning rate is $3\times 10^{-4}$; total number of epochs $400$.
                   \end{itemize} 
                \item \gls{vvCV} (estimating B/fixed B) 
                    \begin{itemize}
                        \item Sample size:  $(40,40)$ from $(\text{Normal}(0,1),    \text{Normal}(0,1))$ for $(f_L, f_H)$.
                        \item Base kernel: first order polynomial kernel.
                        \item Optimisation: When $B$ is fixed, we set $B_{11} = B_{22} = 0.5, B_{12} = B_{21} = 0.01$; otherwise, $B^{(0)}$ is initialized at the identity matrix $I_2$. $\lambda = 10^{-5}$; batch size is  $5$ ($10$ in total for $(f_L, f_H)$); learning rate is $3\times 10^{-4}$; total number of epochs $400$.
                    \end{itemize}
                
            \end{itemize}

\end{itemize}

The empirical computational cost for all tasks is as follows. Scalar-valued CVs take approximately $2.4$ seconds for either choice of kernels; \gls{vvCV}s with fixed $B$ take around $3.3$ seconds with a squared-exponential kernel or around $3.1$ seconds with a $1$st order polynomial kernel; \gls{vvCV}s with estimated $B$ take around $6.6$ seconds with a squared-exponential kernel or around $6.2$ seconds with a $1$st order polynomial kernel.

\subsection{Experimental Details of the Physical Modelling (Borehole) of Waterflow}

\label{appendix:multifidelity_modelling}

\begin{table*}[t!]
    \centering
        \caption{\emph{Prior Distributions for the inputs of the Borehole function}.}
        \label{tab:distributions_in_Borehole_function}
    \begin{tabular}{ c  c || c c}
    \toprule
        Random variable &  Distributions & Random variable &  Distributions  \\ \hline \hline
         $r_w $ & $\text{Normal}(0.1, 0.0161812^2)$  & $r$   &  $\text{Normal}(100, 0.01)$ \\ \hline
         $T_u$  &  $\text{Normal}(89335, 20)$  & $T_l$ &   $\text{Normal}(89.55, 1)$ \\ \hline
         $H_u$  &   $\text{Normal}(1050, 1)$ & $H_l$ &  $\text{Normal}(760, 1)$  \\ \hline
         $L$    &   $\text{Normal}(1400, 10)$ & $K_w$ &   $\text{Normal}(10950, 30)$ \\ 
    \bottomrule
    \end{tabular}
\end{table*}

\noindent In this section, we provide details on the Borehole example from the main paper, and provide complementary experiments. The distributions with respect to which the integral is taken is an eight-dimensional Gaussian with independent marginals provided in Table \ref{tab:distributions_in_Borehole_function}.
 The low-fidelity model and high-fidelity model of water flow \cite{Xiong2013} is given by,
\begin{talign*}
&f_L(x) = \frac{5T_u(H_u-H_l)}{\log\left(\frac{r}{r_w}\right)\Big(1.5+ \frac{2LT_u}{\log\left(\frac{r}{r_w}\right) r_w^2 K_w} + \frac{T_u}{T_l}\Big)}\\
&f_H(x) = \frac{2\pi T_u (H_u - H_l)}{\log\left(\frac{r}{r_w}\right)\Big(1 + \frac{2LT_u}{\log\left(\frac{r}{r_w}\right) r_w^2 K_w} +\frac{T_u}{T_l} \Big) }.
\end{talign*}
 where $x = (r_w, r, T_u, T_l, H_u, H_l, L, K_w)$.
 
\subsubsection{Experiment in the main text: Balanced vv-CVs}
The number of replications is $100$ for all methods. Details of their implementation is given below:
\begin{itemize}
    \item Base kernel:  Instead of using $k(x, x') = \exp(-\|x - x'\|_2^2/2\nu)$ with $l>0$ which implicitly assumes that the length-scales are identical in all directions, we now allow that each dimension can have its own length-scale. That is,
\begin{talign*}
    k(x, x') := \prod_{j=1}^d k_j(x_j, x'_j) \qquad \text{ where } \qquad k_j(x_j, x'_j) = \exp\left(-\frac{ (x_j-x'_j)_2^2}{2 \nu_j}\right). 
\end{talign*}
Each of the components has its own length-scale $\nu_j > 0$ to be determined.
         \item Since $\pi(x) = \prod_{j=1}^d \pi_j(x_j)$, the score function is $           \nabla_x \log \pi(x) = \left( \frac{\partial \log \pi_1(x)}{\partial x_1}, \ldots,
                      \frac{\partial \log \pi_d(x)}{\partial x_d} \right)^\top.$
     \item Hyper-parameter tuning: batch size $5$ ($10$ in total for $(f_L, f_H)$); learning rate of tuning $0.05$; epochs of tuning $20$.
    \item Optimisation (estimated B/pre-fixing B): When $B$ is fixed, we set $B_{11} = B_{22} = 5e-4, B_{12} = B_{21} = 5e-5$; otherwise, $B^{(0)}$ is initialized at $1e-5 \times I_2$. $\lambda = 1e-5$; batch size $5$ ($10$ in total for $(f_L, f_H)$); learning rate for the cases when sample sizes are $(10, 20, 50, 100, 150)$ are $(0.09, 0.06, 0.012, 0.0035, 0.002)$, respectively.
\end{itemize}

\noindent  The empirical computational cost (measured in seconds) of this example is: when $m=(10,10)$, it takes CF, \gls{vvCV}s with fixed B and \gls{vvCV}s with estimated B around $0.03$, $1.2$ and $2.6$ seconds, respectively; when $m=(20,20)$, it takes CF, \gls{vvCV}s with fixed B and \gls{vvCV}s with estimated B around $0.1$, $3$ and $5$ seconds, respectively;  when $m=(50,50)$, it takes CF, \gls{vvCV}s with fixed B and \gls{vvCV}s with estimated B around $0.7$, $7.5$ and $13.5$ seconds, respectively;  when $m=(100,100)$, it takes CF, \gls{vvCV}s with fixed B and \gls{vvCV}s with estimated B around $2.7$, $ 17$ and $27.6$ seconds, respectively; when $m=(150,150)$, it takes CF, \gls{vvCV}s with fixed B and \gls{vvCV}s with estimated B around $6$, $29$ and $49$ seconds, respectively.

\subsubsection{Additional Experiment: Unbalanced Data-sets for Physical Modelling of Waterflow}
\label{appendix:unbalanced_vvCV_borehole}
In Figure \ref{fig:vvCV_unbalanced_cases_Borehole}, we present the results of vv-CVs when the sample sizes are unbalanced; that is, we have a different number of samples for the low-fidelity and high-fidelity models. The exact setup is given below, and we replicated the experiment $100$ times.
\begin{itemize}
    \item Sample size: $m_H$ is fixed to be $20$, while $m_L \in \{20, 40, 60\}$.
    \item Base kernel: product of squared exponential kernels. $k(x, x') := \prod_{j=1}^d k_j(x_j, x'_j)$, where each $k_j(x_j, x'_j) = \exp(-(x_j-x'_j)_2^2/ 2 \nu_j)$ has its own length-scale $\nu_j > 0$ to be determined.
    \item Hyperparameter tuning: batch size of tuning $5$ ($10$ in total for $(f_L, f_H)$); learning rate of tuning is $0.05$; epochs of tuning is $20$.
    \item Optimisation (estimated B/pre-fixing B): When $B$ is fixed, we set $B_{11} = B_{22} = 5\times 10^{-4}, B_{12} = B_{21} = 5\times 10^{-5}$; otherwise, $B^{(0)}$ is initialized at $10^{-5} \times I_2$. $\lambda = 10^{-5}$; learning rate is $(0.06, 0.04, 0.02)$ when $m_L \in \{20, 40, 60\}$, respectively.
\end{itemize}

Interestingly, we notice that not much is gained when increasing the number of samples for the low-fidelity model. In fact, in the case of a fixed $B$, the performance tends to decrease with a larger $m_L$. This is likely due to the ``negative transfer'' phenomenon which is well-known in machine learning. This phenomenon can occur when two tasks are not similar enough to provide any gains in accuracy. In this case, there is clearly no advantage in using a larger $m_L$ since this increases computational cost and does not provide any gains in accuracy.

\begin{figure}
    \centering
    \includegraphics[width=0.7\linewidth]{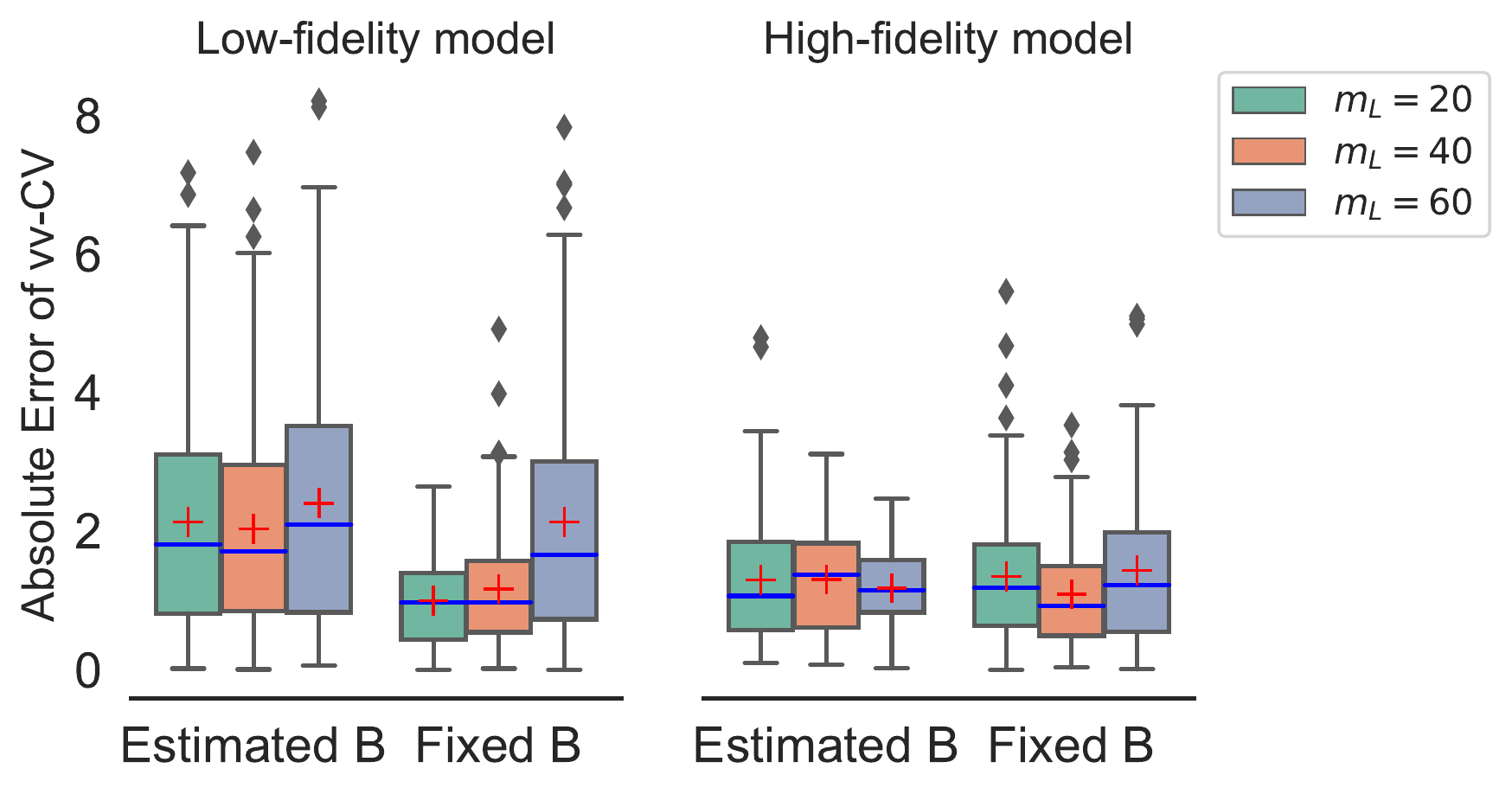}
    \caption{Performance of vv-CVs with unbalanced sample sizes. Here we fix $m_H = 20$, and changing $m_L$ to be $20, 40, 60$. Each experiment is repeated 100 times. }
    \label{fig:vvCV_unbalanced_cases_Borehole}
\end{figure}

\subsection{Experimental Details of the Computation of the Model Evidence through Thermodynamic Integration}\label{appendix:TI}

To implement our \gls{vvCV}s, we need to derive the corresponding score functions.
For a power posterior, the score function is of the form:
  \begin{talign*}
        \nabla_\theta \log p(\theta|y,t) =  t \nabla_\theta \log p(y|\theta) + \nabla_\theta \log p(\theta)
  \end{talign*}
where $\nabla_\theta \log p(\theta)$ is the score function corresponding to the prior. In our case, the prior is a log-normal distribution $\log \theta \sim \mathcal{N}(\mu, \sigma^2)$ (where $\sigma=0.25$), and its score function is given by:
\begin{talign*}
   \nabla_\theta  \log p(\theta) &= -\frac{1}{\theta} - \frac{\log \theta - \mu}{x \sigma^2}.
\end{talign*} 
The score functions for all temperatures are plotted in \Cref{fig:temperature_ladder}; as observed, temperatures consecutive score functions are very similar to one another.

In order to keep the computational cost manageable, we split the $T=62$ integration problems into groups of closely related problems. In particular, we jointly estimate the means in terms of $4$ consecutive temperatures on the ladder ( group 1 is $\mu_1, \mu_2, \mu_3, \mu_4$, group 2 is $\mu_5, \mu_6, \mu_7, \mu_8$, etc...). Since $31$ is not divisible by $4$, our last group consists of three means $\mu_{29}, \mu_{30}, \mu_{31}$. Then, the same approach is taken to create groups of 4 (or 3 for the last group) variances. 

The number of replications was $20$ for each method. Details are given below:
\begin{itemize}
    \item CV
          \begin{itemize}
                   \item[*] Base kernel: Preconditioned squared-exponential kernel \citep{oates2017_CF_for_MonteCarloIntegration}.
                   \item[*] Hyperparameter tuning: we use the values $(0.1, 3)$ in \citep{oates2017_CF_for_MonteCarloIntegration}.
                  \item[*] Optimisation: $\lambda =10^{-3}$; batch size is $5$; total number of epochs is $400$.
          \end{itemize}
    \item vv-CV(estimated B)
          \begin{itemize}
              \item[*] Base kernel: Preconditioned squared-exponential kernel \citep{oates2017_CF_for_MonteCarloIntegration}.
               \item[*] Hyperparameter tuning: we use the values $(0.1, 3)$ in \citep{oates2017_CF_for_MonteCarloIntegration}.
               \item[*] Optimisation: $\lambda = 10^{-3}$; batch size is $5$; learning rate is $0.01$; number of epochs is $400$.
          \end{itemize}
\end{itemize}

\begin{figure}
    \centering
    \includegraphics[width=0.8\linewidth]{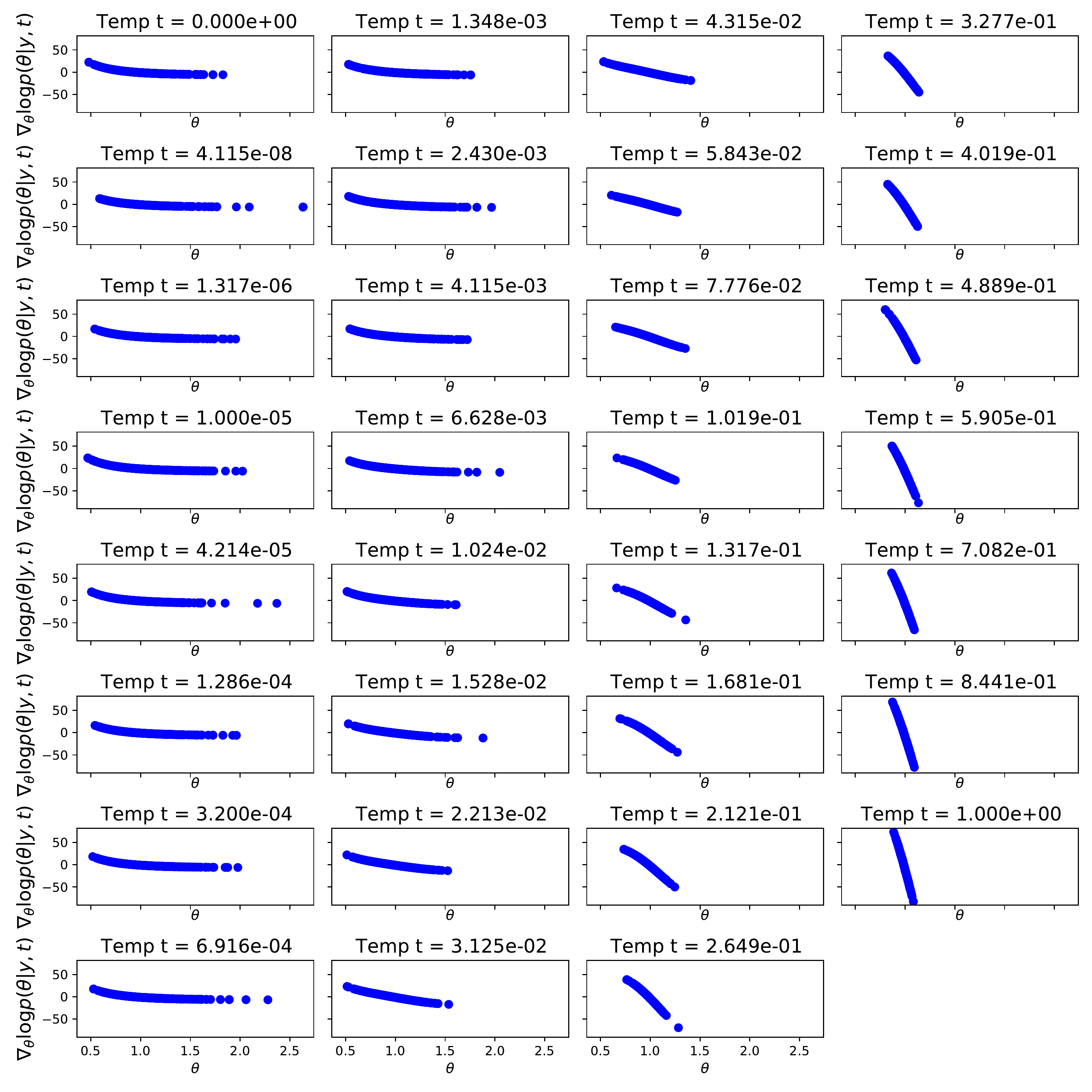}
    \caption{Score functions corresponding to the power posteriors at different temperatures on the temperature ladder.}
    \label{fig:temperature_ladder}
\end{figure}

\subsection{Experimental Details of the Lotka-Volterra System}
\label{appendix:Lotka}
We implement $\log$-$\exp$ transform on model parameters and avoid constrained parameters on the ODE directly. Lotka—Volterra system can be re-parameterized as,
\begin{talign*}
    \frac{dv_1(s)}{ds} &= \tilde{\alpha} v_1(s)- \tilde{\beta} v_1(s) v_2(s)\\
    \frac{dv_2(s)}{ds} &= \tilde{\delta} v_1(s) v_2(s)- \tilde{\gamma} v_2(s),
\end{talign*}
where 
\begin{talign*}
    &\tilde{\alpha} =\exp(\alpha), \tilde{\beta}  = \exp(\beta) , \\
    &\tilde{\delta} = \exp(\delta),  \tilde{\gamma} = \exp(\gamma),
\end{talign*}
where $v_1$ and $v_2$ represents the number of preys and predators, respectively.

\noindent The model is,
\begin{talign*}
    y_{10} \sim \text{Log-Normal}(\log \tilde{v}_1(0), \tilde{\sigma}_{y_1})\\
    y_{20} \sim \text{Log-Normal}(\log \tilde{v}_2(0), \tilde{\sigma}_{y_2}) \\
    y_{1s} \sim \text{Log-Normal}(\log v_1(s), \tilde{\sigma}_{y_1}) \\
    y_{2s} \sim \text{Log-Normal}(\log v_2(s), \tilde{\sigma}_{y_2}) 
\end{talign*}
where 
\begin{talign*}
    &\tilde{v}_1(0): = \exp(v_1(0)), \tilde{v}_2(0): = v_2(0)\\
    &\tilde{\sigma}_{y_1}:=\exp(\sigma_{y_1}), \tilde{\sigma}_{y_2}=\exp(\sigma_{y_2}). 
\end{talign*}

\noindent By doing so, $x :=(\alpha, \beta, \delta, \gamma, v_1(0), v_2(0), \sigma_x, \sigma_y)^\top$ can be defined on the whole $\R^8$ as the exponential transformation will make sure they larger than zero, and thus can be assigned priors on $\R^8$, e.g., Gaussian. As a result, the expectations asscoiated with $\pi(x)$ are defined on $\R^8$ and Stan will return the scores of these parameters directly as these 8 parameters $x$ themselves are unconstrained through manually reparameterisation.

\noindent Priors are,
\begin{talign*}
  \alpha , \gamma &\sim \text{Normal}(0, 0.5^2) \\
  \beta , \delta &\sim \text{Normal}(-3, 0.5^2) \\
  \sigma_x, \sigma_y &\sim \text{Normal}(-1,1^2)\\
  v_1(0), v_2(0) &\sim \text{Normal}(\log 10,1^2)
\end{talign*}

\noindent The fitting for predators $y_{1s}$ and $v_1(s)$ at points $s_1, \ldots, s_m$ are shown in \Cref{fig:lotka_fig}. The fitting for predators $y_{2s}$ and $v_2(s)$ at points $s_1, \ldots, s_m$ are shown in \Cref{fig:lotka_fig_appendix}.

\begin{figure}[ht!]
\vskip 0.2in
\begin{center}
\centerline{\includegraphics[width=0.7\columnwidth]{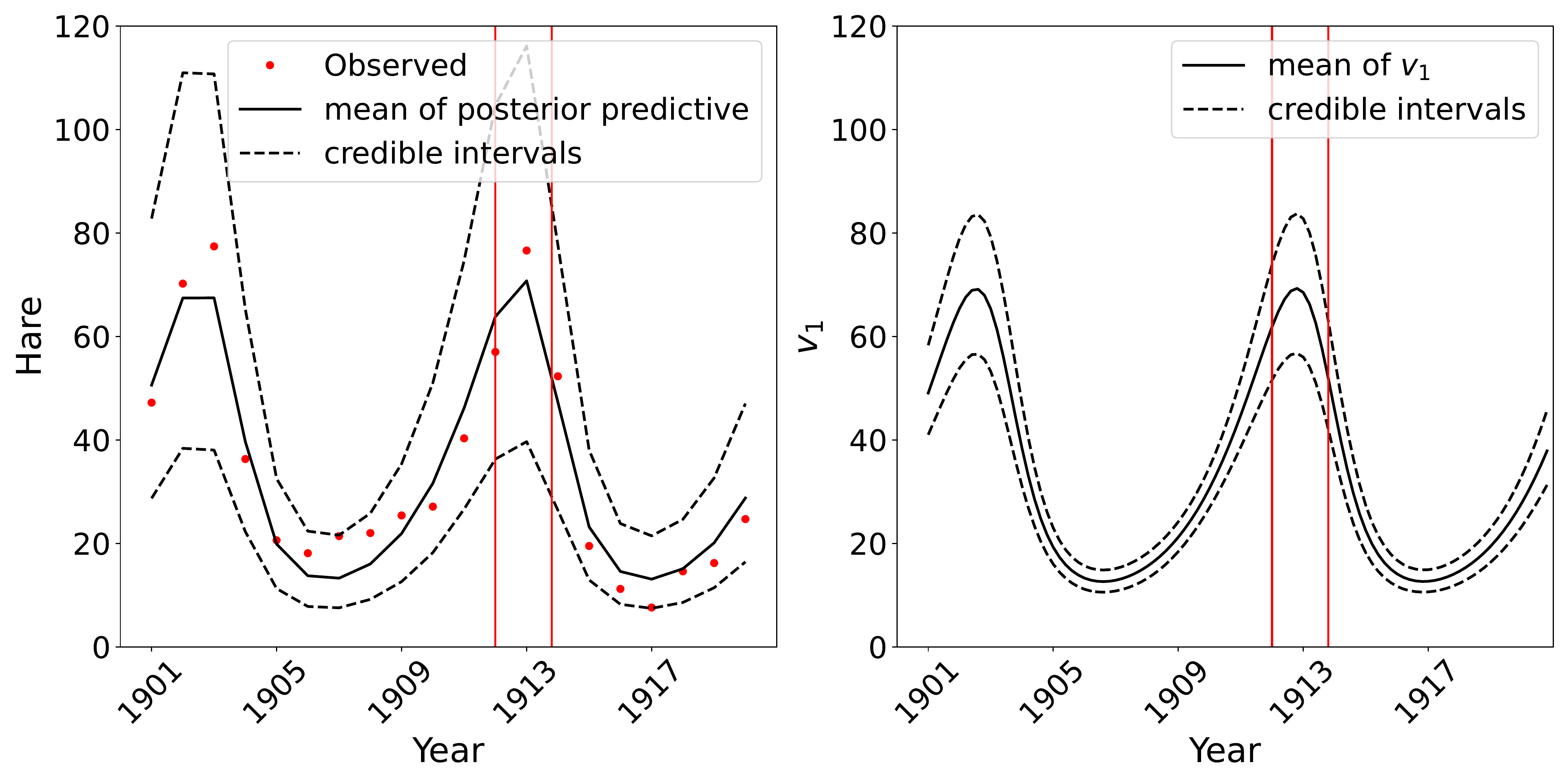}}
\caption{\emph{Bayesian inference of abundance of preys of Lotka-Volterra system.} Dots are observations; lines are the posterior means while dotted lines are the corresponding $95\%$ credible intervals. Tasks are chosen in the area between the two vertical red lines.}
   \label{fig:lotka_fig}
\end{center}
\vskip -0.2in
\end{figure}

\begin{figure}[ht!]
    \centering
    \includegraphics[width=0.7\columnwidth]{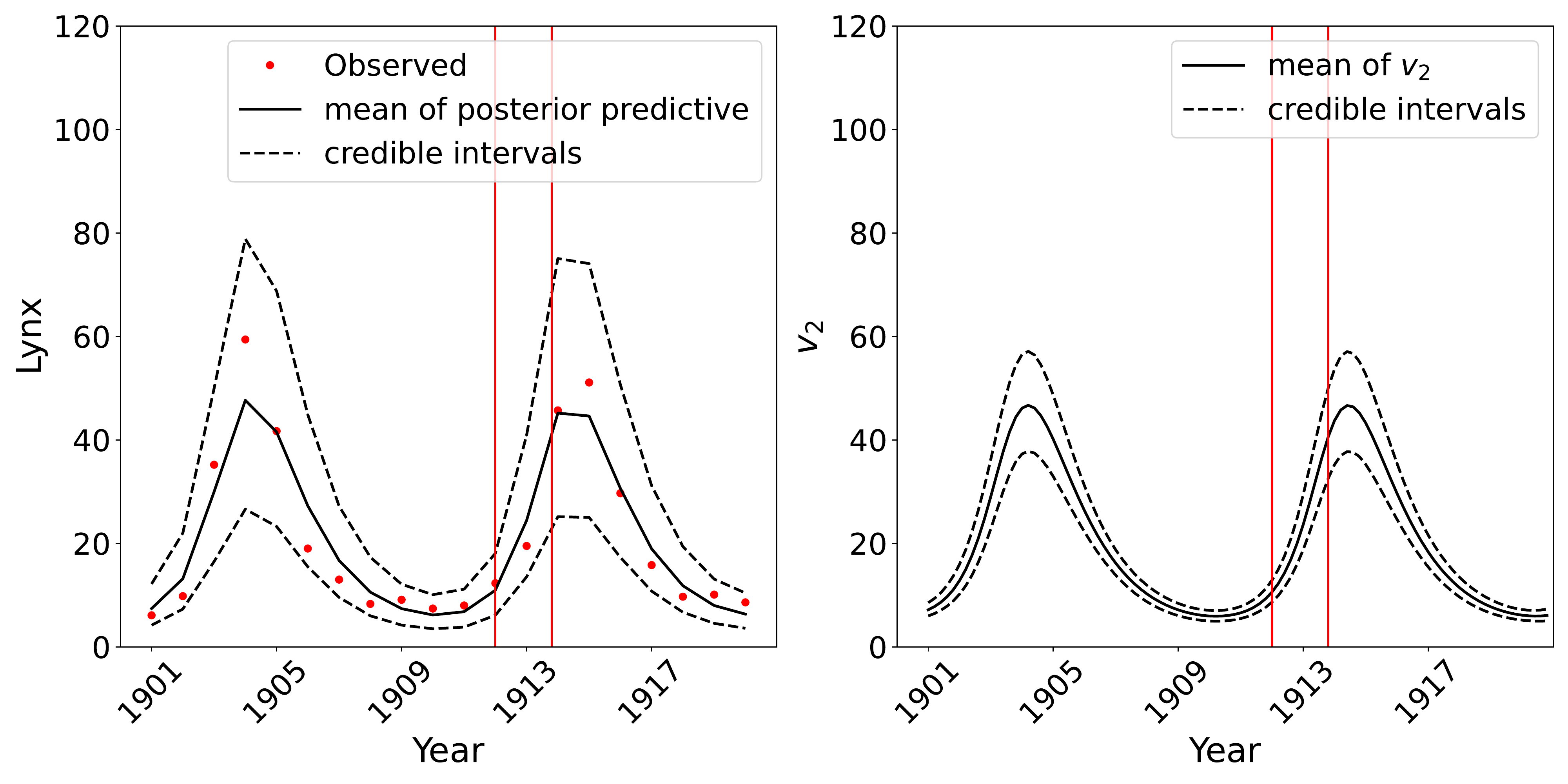}
    \caption{\emph{Bayesian inference of abundance of predators of Lotka-Volterra system.} Dots are observations; lines are the posterior means while dotted lines are the corresponding $95\%$ credible intervals. Tasks are chosen in the area between the two vertical red lines.}
    \label{fig:lotka_fig_appendix}
\end{figure}

\noindent The number of replications is $10$ for each method and for each task. Details are given below:
\begin{itemize}
    \item Tasks $\Pi[f_t]$ at time $s_1', \ldots, s_T'$ (the base unit is $1$ year):
    \begin{itemize}
        \item[*] $T=2$:  $1913. , 1913.2$;
        \item[*] $T=5$: $1912. , 1912.2, 1912.4, 1912.6, 1912.8$;
        \item[*] $T=10$: $1912. , 1912.2, 1912.4, 1912.6, 1912.8, 1913. , 1913.2, 1913.4, 1913.6, 1913.8$.
    \end{itemize}
    \item Base kernel: same one as in \ref{appendix:multifidelity_modelling}: product of squared-exponential kernels. 
    \item Hyperparameter tuning: batch size is $10$; learning rate $0.01$; total number of epochs $10$.
    \item Optimisation: $\lambda= 10^{-5}$; batch size $10$; learning rate $10^{-3}$; total number of epochs $400$.
\end{itemize}


\noindent The empirical computational cost (measured in seconds) of this example is: when $T=2$, it takes CF, scalar-valued CVs, \gls{vvCV}s with fixed B and \gls{vvCV}s with estimated B around $67$, $80$, $150$ and $159$ seconds respectively for all $T$ tasks; when $T=5$, it takes CF, scalar-valued CVs, \gls{vvCV}s with fixed B and \gls{vvCV}s with estimated B around $170$, $200$, $854$ and $882$ seconds respectively for all $T$ tasks; when $T=10$, it takes CF, scalar-valued CVs, \gls{vvCV}s with fixed B and \gls{vvCV}s with estimated B around $340$, $400$, $3435$ and $3469$ seconds respectively for all $T$ tasks.

\subsubsection{Addition Experiments for Bayesian inference of abundance of preys of Lotka-Volterra system}
We present additional experiments in \Cref{tab:Lotka1_additional} under the same settings as those in the above section. We consider to estimate $\Pi[f_t]$ at time $s_1', \ldots, s_T'$ (the base unit is $1$ year),
\begin{itemize}
    \item $T=2$: $1915., 1915.2$;
    \item $T=5$: $1915. , 1915.2, 1915.4, 1915.6, 1915.8$;
    \item $T=10$: $1914. , 1914.2, 1914.4, 1914.6, 1914.8, 1915. , 1915.2, 1915.4, 1915.6, 1915.8$.
\end{itemize}

\begin{table*}[ht!]
    \centering
    \caption{\emph{Additional Experiments for the Bayesian inference of the Lotka-Volterra system: Sum of mean absolute error of each task.}}
    \label{tab:Lotka1_additional}
    \begin{tabular}{c  c|| cc  cc cc cc}
    \toprule
      $T$  & $m$ &  \gls{vvCV}- Estimated B  & \gls{vvCV}-Fixed B   & CF & MC\\ \hline \hline
        2& $500$ & 0.169  & 0.144  & 0.242  &  0.275   \\        \hline
        5 & $500$  &  0.322 & 0.245  & 1.246 &  0.661    \\ \hline
        10 & $500$  &  0.916 & 0.792 & 5.835 & 1.797  \\
    \bottomrule
    \end{tabular}
\end{table*}



\end{document}